\documentclass[11pt]{paper}

\usepackage{latexsym}
\usepackage{fullpage}
\usepackage{epsfig}
\usepackage{verbatim}
\usepackage{cite}
\usepackage{amssymb}
\usepackage{amsmath}
\usepackage{amsthm}
\usepackage{graphicx}
\usepackage{color}

\newtheorem{theorem}{Theorem}[section]

\newtheorem{lemma}[theorem]{Lemma}

\newtheorem{claim}[theorem]{Claim}

\newtheorem{prop}[theorem]{Proposition}
\newtheorem{corol}[theorem]{Corollary}

\DeclareMathOperator{\conv}{CH}
\DeclareMathOperator{\UH}{UH}
\DeclareMathOperator{\LH}{LH}
\DeclareMathOperator{\lev}{lev}

\newcommand{\eqdef}{:=}

\DeclareMathOperator{\EX}{\mathbf{Exp}}

\newcommand\A{{\mathcal A}}

\newcommand\T{{\mathcal T}}

\newcommand\etal{\textit{et~al.}}

\newcommand\R{{\mathbb R}}
\newcommand\tk{\tilde{k}}

\usepackage{graphicx,floatflt,psfrag}

\title{Convex Hull of Points Lying on Lines in {$o(n \log{n})$} Time after Preprocessing\thanks
  {A preliminary version appeared as
  E. Ezra and W. Mulzer.  \emph{Convex Hull of Imprecise Points in $o(n log n)$ 
  Time after Preprocessing} in Proc.~27th SoCG, pp.~11--20, 2011}}

\author{Esther Ezra\thanks{%
    Supported by a PSC-CUNY Research Award.}
  \and
Wolfgang Mulzer\thanks{%
Supported in part by NSF grant CCF-0634958, 
    NSF CCF 083279, and a Wallace Memorial Fellowship in Engineering.}
}

\institution{
    Courant Institute \\
    New York University\\ 
    New York, USA\\
    \texttt{esther@cims.nyu.edu}
 \and 
   Institut f\"ur Informatik\\
 Freie Universit\"at Berlin\\
 Berlin, Germany\\
\texttt{mulzer@inf.fu-berlin.de}
 }

\begin{document}

\maketitle

\begin{abstract}
  Motivated by the desire to cope with \emph{data imprecision}~\cite{Loeffler09}, 
  we study methods for taking advantage of
  preliminary information about point sets in order to speed
  up the computation of certain structures associated with them.

  In particular, we study the following problem: 
  given a set $L$ of $n$
  lines in the plane, we wish to preprocess $L$ such that later, upon receiving
  a set $P$ of $n$ points, each of which lies on a distinct line of $L$, we can
  construct the convex hull of $P$ efficiently. We show that in quadratic time
  and space it is possible to construct a data structure on $L$ that enables us
  to compute the convex hull of any such point set $P$ in $O(n \alpha(n) \log^*
  n)$ expected time. If we further assume that the points are ``oblivious''
  with respect to the data structure, the running time improves to $O(n
  \alpha(n))$. 
  The same result holds when $L$ is a set of line segments (in general position).
  We present several
  extensions, including a trade-off between space and query time and an
  output-sensitive algorithm. We also study the ``dual problem'' where we 
  show how to efficiently compute the $(\leq k)$-level of $n$ lines in the 
  plane, each of which is incident to a distinct point (given in advance).

  We complement our results by $\Omega(n \log{n})$ lower bounds under the
  algebraic computation tree model for several related problems, including 
  sorting a set of points (according to, say, their $x$-order), each of which 
  lies on a given line known in advance. Therefore, the convex hull problem 
  under our setting is \emph{easier} than sorting, contrary to the 
  ``standard'' convex hull and sorting problems, in which the two problems 
  require $\Theta(n \log{n})$ steps in the worst case (under the algebraic 
  computation tree model).

\end{abstract}


\section{Introduction}
Most studies in computational geometry rely on an unspoken assumption:
whenever we are given a set of input points, their precise 
locations are available to us. 
Nowadays, however, the input is often obtained via sensors from
the real world, and hence it comes with an inherent imprecision. 
Accordingly, an increasing effort is being devoted to achieving a better 
understanding of data imprecision and to developing tools to cope with it 
(see, e.g.,~\cite{Loeffler09} and the references therein). 
The notion of imprecise data can be formalized in numerous 
ways~\cite{Loeffler09,GuSalesinSt89,LofflerPhillips09}.
We consider a particular setting that has recently attracted considerable 
attention~\cite{BuchinLoMoMu11,Devillers11,HeldMi08,vKreveldLoMi10,LoefflerMu11,LoefflerSn10}.
We are given a set of planar regions, each of which represents an 
estimate about an input point, and the exact coordinates of the points 
arrive some time later and need to be processed quickly.
This situation could occur, e.g., during a two-phase measuring process: 
first the sensors quickly obtain a rough estimate of the data, 
and then they invest considerably more time to find the precise locations.
This raises the necessity to preprocess the
preliminary (imprecise) locations of the points, and store them in an 
appropriate data structure, so that when the exact measurements of the 
points arrive 
we can efficiently compute a pre-specified structure on them. 
In settings of this kind, we assume that for each input point 
its corresponding region is known (note that by this assumption we also 
avoid a point-location overhead). In light of the applications,
this is a reasonable assumption, and it can be implemented by, e.g.,
encoding this information in the ordering of $P$.

\subsection*{Related work}

\noindent\textbf{Data imprecision.}
Previous work has mainly focused on computing a triangulation for the 
input points. 
Held and Mitchell~\cite{HeldMi08} were the first to consider this framework,
and they obtained optimal bounds for preprocessing disjoint unit disks for
point set triangulations, a result that was later generalized
by van Kreveld~\etal~\cite{vKreveldLoMi10} to arbitrary disjoint polygonal 
regions. For \emph{Delaunay} triangulations, L{\"o}ffler and 
Snoeyink~\cite{LoefflerSn10} obtained an optimal result for disjoint unit 
disks (see also~\cite{Devillers11, LoefflerMu11}), which was later 
simplified and generalized by 
Buchin~\etal~\cite{BuchinLoMoMu11} to \emph{fat}\footnote{
  A planar region $o$ is said to be fat if there exist two concentric disks,
  $D \subseteq o \subseteq D'$, such that the ratio between the radii of 
  $D'$ and $D$ is bounded by some constant.}
and possibly intersecting regions.
If $n$ is the number of input regions,
the preprocessing phase typically takes $O(n \log n)$ time and 
yields a linear size data structure; the time to find the structure
on the exact points 
is usually linear or depends on the complexity (and the fatness)
of the input regions.

Since the convex hull can be easily extracted from the Delaunay 
triangulation in linear time, the same bounds carry over. However,
once the regions are not necessarily fat, the techniques 
in~\cite{BuchinLoMoMu11, LoefflerSn10} do not yield the aforementioned 
bounds anymore. In particular, if the regions consist of lines or 
line segments, one cannot hope (under certain computational models) to 
construct the Delaunay triangulation of $P$ in time $o(n \log n)$, 
regardless of preprocessing (see~\cite{DjidjevLi95} and 
Section~\ref{sec:lowerbounds}). 
Nevertheless, if we are less ambitious 
and just wish to compute  the \emph{convex hull} of $P$, we can achieve 
better performance, as our main result shows.

\noindent\textbf{Convex hull.}
Computing the convex hull of a planar $n$-point set is perhaps the most
fundamental problem in computational geometry, and there are many
algorithms available~\cite{deberg2008cga,PS-85}.
All these algorithms require $\Theta(n \log n)$ steps,
which is optimal in the algebraic computation tree model~\cite{BenOr83}.
However, there are numerous ways to exploit additional information
to improve this bound. For example, if the points are sorted along any 
fixed direction, Graham's scan takes only linear time~\cite{deberg2008cga}. 
If we know that
there are only $h$ points on the hull, the running time reduces to 
$O(n\log h)$~\cite{KirkpatrickSe86, AfshaniBaCh09}. 
If the points constitute the vertices of a given polygonal chain, 
the complexity again reduces to linear~\cite{McCallumA79}.
Our work shows another setting in which additional information
can be used to circumvent the theoretic lower bound.

Another somewhat related problem (albeit conceptually different)
is the \emph{kinetic convex hull} problem, where we are given $n$ 
points which move \emph{continuously} in the plane, and the goal 
is to maintain their convex hull over time. Kinetic data structures have been 
introduced by Basch~\etal~\cite{BGH-99} and received considerable attention 
in follow-up studies (see, e.g.,~\cite{AD-07} and the references therein).
When the trajectories of the points are lines, our problem can 
be interpreted as a (perhaps, extended and intricate) variant of the 
kinetic convex hull problem. 
Indeed, if the goal is to preprocess the linear trajectories such that 
the convex hull can be reported efficiently at any given time $t$,
our algorithm applies (in which case the exact set of points $P$ consists 
of their
positions at time $t$) and yields a relatively simple solution.
Nevertheless, our problem is more intricate than the kinetic convex hull
problem for linear trajectories, as in our scenario there is no continuous
motion that enables us to have a better control on the exact set of points
(once they arrive).

\noindent\textbf{Our results.}
We show that under a mild assumption (see Section~\ref{sec:oblivious}) 
we can preprocess the input lines $L$ such that 
given any set $P$ of points, each of which lies on a distinct line of $L$,
the convex hull $\conv(P)$ can be computed in expected time $O(n\alpha(n))$, 
where
$\alpha(\cdot)$ is the (slowly growing) inverse Ackermann 
function~\cite[Chapter~2.1]{SA-95}; the expected running time
is $O(n\alpha(n)\log^{*}n)$ without this assumption.
Our data structure has quadratic preprocessing time and storage,
and the convex hull algorithm is based on a batched randomized 
incremental construction similar to Seidel's tracing technique~\cite{Seidel91}.
As part of the construction, we repeatedly trace the \emph{zone} 
of (the boundary of) an intermediate hull in the \emph{arrangement} 
of the input lines (see below for the definitions).
The fact that the complexity of the zone is only 
$O(n\alpha(n))$~\cite{BernEpPlYa91,SA-95}, and that it can be computed in 
the same asymptotic time bound 
(after having the arrangement at hand), is a key property of our solution.
The analysis also applies 
when $L$ is a set of line segments, 
and yields the same result.

We also show that the analogous 
problem in which we just wish to sort the points according
to their $x$-order imposes algebraic computation trees of depth 
$\Omega(n \log{n})$. Hence, in our 
setting convex hull computation is strictly easier than sorting, contrary
to the ``standard'' (unconstrained) model, in which both problems
are equivalent in terms of hardness (see, e.g.,~\cite{deberg2008cga}). 
Our results can be extended with similar bounds to several related problems, 
such as determining the width and diameter of $P$, as well as time-space 
trade-offs and designing an output-sensitive algorithm. 
Unfortunately, already for the closest pair problem a 
preprocessing of the regions is unlikely to decrease the query time to 
$o(n \log{n})$, demonstrating once again the delicate nature of 
our setting.

In Section~\ref{sec:levels} we study 
a generalization of the problem under the dual setting.
Specifically, we wish to preprocess a planar $n$-point set $P$ 
such that given an integer $k$ and a set $L$ of lines, each of 
which is incident to a distinct point of $P$, we can find the 
``$(\leq k)$-level'' in the arrangement of $L$ efficiently. 
We show a randomized construction whose expected running time is 
$O(n\alpha(n) + nk)$ under a mild assumption, and 
$O(n\alpha(n)\log^{*}n + nk)$ without this assumption. 
As above, our data structure has quadratic preprocessing time and storage. 
This improves over the $O(n \log n + nk)$ time algorithms in the traditional 
model~\cite{Chan00,EverettRovKr96}, as long as $k = o(\log n)$.
Our approach is a non-trivial extension of the technique presented 
in Section~\ref{sec:quadratic_DS}, incorporated with the algorithms of 
Chan~\cite{Chan00} and Everett~\etal~\cite{EverettRovKr96}, as well as 
the Clarkson-Shor technique~\cite{CS-89}.

The quadratic preprocessing time and storage might seem disappointing.
However, a related lower bound by Ali Abam and de Berg~\cite{AD-07} from
the study of kinetic convex hulls (albeit providing a weaker evidence)
suggests that quadratic space might be necessary, and that only relatively weak 
time-space trade-offs (as in Section~\ref{sec:extensions}) are possible in this
model (see the discussion in Section~\ref{sec:extensions} for further details). 
Given the hardness of related problems, and the fact that previous approaches 
fail for ``thin'' regions, it still seems remarkable that improved bounds are 
even possible.

\section{Convex Hulls}
\label{sec:quadratic_DS}

\noindent\textbf{Preliminaries.}
The input at the preprocessing stage is a set $L$ of $n$ lines in the plane. 
A \emph{query} to the resulting data structure consists of any point set $P$ 
such that each point lies on a distinct line in $L$,
and for every point we are given its corresponding line.
For simplicity, and without loss of generality, we assume
that both $L$ and $P$ are in \emph{general position} 
(see, e.g.,~\cite{deberg2008cga,SA-95}).
We denote by $\conv(P)$ the convex hull of $P$, and by
$E(P)$ the edges of $\conv(P)$. 
We represent the vertices of $\conv(P)$ in clockwise order, and 
we direct each edge $e \in E(P)$ such that $\conv(P)$ lies to its right.
Given a subset $Q \subset P$, a point $p \in P \setminus Q$, 
and an edge $e \in E(Q)$, we say that $e$ is in \emph{conflict} with $p$
if $p$ lies to the left of the line supported by $e$.
The set of all points in $P \setminus Q$ in conflict with
$e$ is called the \emph{conflict list} $C_e$ of $e$, and 
its cardinality is called the \emph{conflict size} $c_e$ of $e$.

In what follows we denote the \emph{arrangement} of $L$ by $\A(L)$,
defined as the decomposition of the plane into \emph{vertices}, 
\emph{edges} and \emph{faces} 
(also called \emph{cells}), each being a maximal connected set contained in the 
intersection of at most two lines of $L$ and not meeting any other line.
The \emph{complexity} of a face $f$ in $\A(L)$ is the number of edges incident 
to
$f$. The \emph{zone} of a curve $\gamma$ consists of all faces that intersect
$\gamma$, and the complexity of the zone is the sum of their complexities.

\subsection{The Construction}
\label{sec:quickCH}
 
\noindent\textbf{Preprocessing.} 
We construct in $O(n^2)$ time (and storage) the \emph{arrangement} $\A(L)$ of
$L$, and produce its \emph{vertical decomposition}, that is, we erect an upward
and a downward vertical ray through each vertex $v$ of $\A(L)$ until they meet
some line of $L$ (not defining $v$), or else extend to infinity.  

\noindent\textbf{Queries.} 
Given an exact point set $P = \{p_1,\ldots, p_n\}$ as 
described above, we obtain $\conv(P)$ through a \emph{batched} 
randomized incremental construction. Let 
$P_1 \subseteq P_2 \subseteq \cdots \subseteq P_{\log^* n} = P$ be
a sequence of subsets, where $P_{k-1}$ is a random sample of $P_{k}$ of size 
$z_{k-1} \eqdef \min\{ \lfloor{n/\log^{(k-1)} n}\rfloor, n\}$,
for $k = 2,\ldots, \log^* n$.\footnote{Here, $\log^{(i)} n$ is the $i$th 
  iterated logarithm: $\log^{(0)} n = n$ and 
  $\log^{(k)} n = \log(\log^{(k-1)} n)$. The standard notation $\log^* n$ is the
  smallest $k$ such that $\log^{(k)} n \leq 1$.}
This sequence of subsets is called a \emph{gradation}.
The idea is to construct $\conv(P_1)$, $\conv(P_2)$,
$\ldots$, $\conv(P_{\log^* n})$ one by one, as follows.
First, we have $|P_1| = O(n/\log{n})$, so it takes $O(n)$ time to
find $\conv(P_1)$, using, e.g., Graham's scan~\cite{deberg2008cga}.
Then, for $k = 2,\ldots, \log^* n$, 
we incrementally construct $\conv(P_k)$ by updating $\conv(P_{k-1})$.
This basic technique was introduced by Seidel~\cite{Seidel91} 
and it has later been exploited by 
several others~\cite{Devillers92,Ramos99,ChazelleMu11}.

To construct $\conv(P_k)$ from $\conv(P_{k-1})$, we use
the data structure from the preprocessing to 
quickly construct the conflict lists of the edges in $E(P_{k-1})$ 
with respect to $P_k$.  In the standard Clarkson-Shor randomized incremental 
construction~\cite{CS-89} it takes $O(n \log n)$ time to maintain
the conflict lists.  However, once we have the 
arrangement $\A(L)$ at hand, this can be done significantly faster.

In fact, we use a refinement of the conflict lists:
we shoot an upward vertical ray from each point on the upper 
hull of $P_{k-1}$, and a downward vertical ray from each point on the
lower hull. Furthermore, we erect vertical walls through 
the leftmost and the rightmost points of $\conv(P_{k-1})$. 
This partitions the complement of $\conv(P_{k-1})$ into
vertical slabs $S(e)$, for each edge $e \in E(P_{k-1})$, and two
boundary slabs $S(v_l)$, $S(v_r)$,
associated with the respective leftmost and rightmost vertices
$v_l$ and $v_r$ of $\conv(P_{k-1})$.
The \emph{refined conflict list} of $e$, $C^{*}_e$, is defined as
$C^{*}_e \eqdef (P_k \setminus P_{k-1}) \cap S(e)$. 
We add to this collection the sets 
$C^{*}_{v_l} \eqdef (P_k \setminus P_{k-1}) \cap S(v_l)$ and
$C^{*}_{v_r} \eqdef (P_k \setminus P_{k-1}) \cap S(v_r)$, which we
call the refined conflict lists of $v_l$ and $v_r$, respectively. 
Note that $C^{*}_e \subseteq C_e$, for every $e \in E(P_{k-1})$.
Moreover, $C^{*}_{v_l}$ (resp., $C^{*}_{v_r}$) is contained in 
$C_{e_1} \cup C_{e_2}$,
where $e_1, e_2 \in E(P_{k-1})$ are the two respective edges 
emanating from $v_l$
(resp., $v_r$);
see Figure~\ref{fig:visible_edge}(a).
We now state a key property of the conflict lists $C_e$
(this property is fairly standard and follows from
related studies~\cite{ChazelleMu11, CS-89, Ramos99}):

\begin{lemma}
  \label{lem:conf_lists}
  Let $Q$ be a planar $m$-point set, $r$ a positive integer satisfying 
  $1 \le r \le m$, and $R \subseteq Q$ a random subset of size $r$. 
  Suppose that $f(\cdot)$ is a monotone non-decreasing function,
  so that $f(x)/x^c$ is decreasing, for some constant $c > 0$.
  Then
  \[ 
  \EX\Bigl[\sum_{e \in E(R)} f(c_e)\Bigr] = 
  O\Bigl( r \cdot f\bigl(m/r\bigr)\Bigr),
  \] 
  where the constant of proportionality depends on $c$, 
  and $c_e$ is the number of points $p \in Q \setminus R$ in conflict 
  with $e \in E(R)$.\qed
\end{lemma}

In other words, the above lemma implies that, on average, the size of the 
conflict list of a fixed edge $e \in E(R)$ is $m/r$ (this can easily be 
seen by setting $f(\cdot)$ to the identity function, and obtaining an overall linear size).  
 
\noindent\textbf{Constructing the refined conflict lists.}
We next present how to construct the refined conflict lists
at the $k$-th round of the algorithm.
We first construct, in a preprocessing step, the refined conflict lists
$C^{*}_{v_l}$, $C^{*}_{v_r}$ in overall $O(z_{k})$ time.
We call these points the \emph{extreme} points,
and for the sake of the analysis, we eliminate these points from $P_{k}$
for the time being, and continue processing them only at the final 
step of the construction---see below.

Let $\UH(P_{k-1})$ be the upper hull of $P_{k-1}$, 
and let $\LH(P_{k-1})$ be its lower hull.
Having these structures at hand, we construct the zones of 
$\UH(P_{k-1})$ and $\LH(P_{k-1})$ in $\A(L)$.
This takes overall $O(n\alpha(n))$ time, using 
the vertical decomposition of $\A(L)$ 
and the fact that the zone complexity of a convex curve in a 
planar arrangement of $n$ lines is $O(n \alpha(n))$; 
see Bern~\etal~\cite{BernEpPlYa91} and 
Sharir and Agarwal~\cite[Theorem~5.11]{SA-95}.

As soon as we have the zones as above,
we can determine for each line $\ell \in L$ the  
edges $e \in E(P_{k-1})$ that $\ell$ intersects (if any). 
Let $L_1$ be the lines that intersect $\conv(P_{k-1})$,
and put $L_2 \eqdef L \setminus L_1$.
(At this stage of the analysis, we ignore all lines corresponding to 
points in $P_k$ that were eliminated at the time we processed 
the extreme points.) 

Next, we wish to find, for each point $p \in P_k \setminus P_{k-1}$
the edges in $E(P_{k-1})$ in conflict with $p$.
If $p$ lies inside $\conv(P_{k-1})$, 
there are no conflicts.
Otherwise, we efficiently find an edge $e_p \in E(P_{k-1})$ 
visible from $p$, whence we search for 
the slab $S(e_p^{*})$ containing $p$---see below.

Let us first consider the points on the lines in $L_1$. 
Fix a line $\ell \in L_1$, let $p \in P$ be the point on $\ell$,
and let $q_1$, $q_2$ be the intersections between $\ell$ and the
boundary of $\conv(P_{k-1})$.
The points $q_1$, $q_2$ subdivide $\ell$ into two
rays $\rho_1$, $\rho_2$, 
and the line segment $\overline{q_1 q_2}$.
By convexity, 
$\overline{q_1 q_2} \subseteq \conv(P_{k-1})$ 
and the rays $\rho_1$, $\rho_2$ 
lie outside $\conv(P_{k-1})$.
Hence, if $p$ lies on $\overline{q_1 q_2}$, it must 
be contained in $\conv(P_k)$.
Otherwise, $p$ sees an edge of $E(P_{k-1})$ that meets one of the 
rays $\rho_1$, $\rho_2$, and we thus set $e_p$ to be this edge 
(which can be determined in constant time); 
see Figure~\ref{fig:visible_edge}(b).
\begin{figure*}
  \begin{center}
    \includegraphics{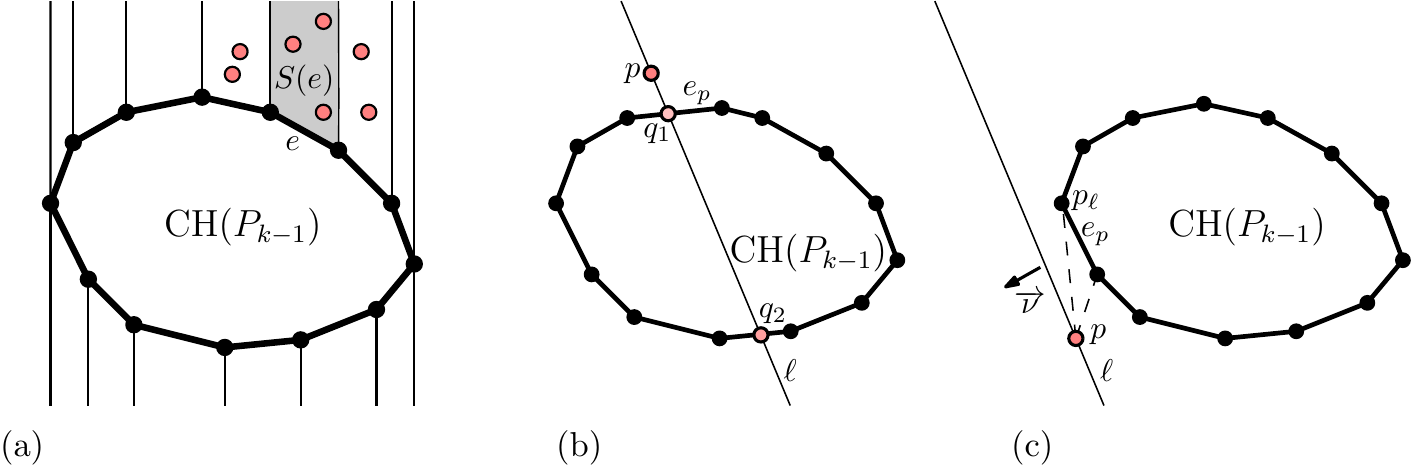}
  \end{center}
  \caption{
  (a) The conflict list $C_e$ of the edge $e \in E(P_{k-1})$ 
      contains all the lightly-shaded 
      points,  
      whereas the refined conflict list $C^{*}(e)$ 
      has only those 
      points in the vertical slab $S(e)$;
      (b--c)   
      The edge $e_p$ of $E(P_{k-1})$ is visible to $p$ when
      (b) $\ell$ intersects $\conv(P_{k-1})$, or (c) $\ell$ does not 
      meet $\conv(P_{k-1})$. In this case $p_{\ell}$ is an extreme vertex
      for the direction $\protect\overrightarrow{\nu}$, 
      and the two dashed lines depict the visibility lines between $p$ 
      and the two respective endpoints of $e_p$.
    }
  \label{fig:visible_edge}
\end{figure*}

We next process the lines in $L_2$. Note that all points on the 
lines in $L_2$ conflict with at least one edge in $E(P_{k-1})$, 
since no line in $L_2$ meets $\conv(P_{k-1})$. To find these edges  
we determine for each $\ell \in L_2$
a vertex $p_\ell$ on the boundary of $\conv(P_{k-1})$
that is extreme for $\ell$.\footnote{By this we mean that $p_\ell$ is
  extremal in the direction of the outer normal 
  of the halfplane that is bounded by $\ell$ and contains $\conv(P_{k-1})$.}
This can be done in total time $O(n)$ by ordering 
$E(P_{k-1})$ 
and $L_2$ according to their slopes 
(the latter being performed during preprocessing), 
and then merging these two lists in linear time.
Next, fix such a line $\ell \in L_2$, and let $p \in \ell$ be a query point,
then $p$ must see one of the two edges in $E(P_{k-1})$
incident to $p_\ell$ (which can be determined in constant time given 
$p_\ell$), and we thus set $e_p$ to be the corresponding edge;
see Figure~\ref{fig:visible_edge}(c).

We are now ready to determine, for each point $p \in  P_{k}$
outside $\conv(P_{k-1})$, the slab $S(e_p^{*})$ that contains it
(note that $e_p^{*}$ must be vertically visible from $p$). 
If $e_p$ is vertically visible from $p$, we set $e_p^{*} := e_p$.
Otherwise, we walk along (the boundary of) $\conv(P_{k-1})$, starting from 
$e_p$ and progressing in the appropriate direction (uniquely determined
by $p$ and $e_p$), until the slab containing $p$ is found.
Using cross pointers between the edges and the points, we can easily 
compute $C^{*}_e$ for each $e \in E(P_{k-1})$.
By construction, all 
traversed edges are in conflict with $p$, 
and thus the overall time for this procedure is proportional to the 
total size of the conflict lists $C_e$. 
Recalling that $c_e = |C_e|$, we obtain
\[
\EX\Bigl[\sum_{e \in  E(P_{k-1})} c_e\Bigr] = O(z_{k}) = O(n),
\]
by Lemma~\ref{lem:conf_lists} with $f: m \mapsto m$.
This concludes the construction of the refined conflict lists.

\noindent\textbf{Computing $\conv(P_k)$.}
We next describe how to construct the upper hull of $P_k$, 
the analysis for the lower hull is analogous.
Let $\langle e_1, \ldots, e_{s}\rangle$
be the edges
along the upper hull of $P_{k-1}$, ordered from left to right. 
For each $e_i$, we sort the points in $C^{*}_{e_i}$ 
according to their $x$-order, using, e.g., merge sort.
We apply the same procedure for the extreme points.
We then concatenate the sorted lists 
$C^{*}_{v_l}, C^{*}_{e_1}, C^{*}_{e_2}, \ldots, C^{*}_{e_s}, C^{*}_{v_r}$,
and merge the result with the vertices of 
the upper hull of $P_{k-1}$. Call the resulting list $Q$, and use
Graham's scan to find the upper hull of $Q$ in time $O(|Q|)$. This
is also the upper hull of $P_{k}$.
Applying once again Lemma~\ref{lem:conf_lists} with $f: m \mapsto m \log m$,
and putting $c^{*}_e \eqdef |C^{*}_{e}|$, 
$c^{*}_{v_l} \eqdef |C^{*}_{v_l}|$,
$c^{*}_{v_r} \eqdef |C^{*}_{v_r}|$, 
and $A > 0$ an absolute constant, 
the overall expected running time of this step is bounded by
\begin{align*}
\label{eq:net_size}
  &\EX\bigl[A \cdot (c^{*}_{v_l} \log{c^{*}_{v_l}} + c^{*}_{v_r} \log{c^{*}_{v_r}} +
    \sum_{e \in E(P_{k-1})} c^{*}_e \log c^{*}_e )\bigr]  \\
 &\leq \EX\bigl[3A \cdot \sum_{e \in E(P_{k-1})} c_e \log c_e\bigr] 
   =
  O\left(z_{k} \log \left(z_{k}/z_{k-1}\right) \right) \\
   &= 
  O\left((n/\log^{(k)} n)
  \log \left(\log^{(k-1)} n/\log^{(k)} n\right)\right) 
   = O(n),
\end{align*}
because by definition $C^{*}_{e} \subseteq C_e$, so $c^*_e \leq c_e$, and
$c^*_{v_l} \leq c_{e_1} + c_{e_2}$ for two edges $e_1$, $e_2$ of $E(P_{k-1})$ (and
similarly for $c^*_{v_r}$).
In total,
we obtain that the expected time to construct 
$\conv(P_{k})$ given $\conv(P_{k-1})$ is $O(n\alpha(n))$,
and since there are $\log^* n$ iterations, 
the total running time is $O(n\alpha(n)\log^* n)$. 

We note that the analysis proceeds almost verbatim when $L$ is just a set 
of \emph{line segments} in the plane.
In this case, we preprocess the lines containing the input segments,
and proceed as in the original problem.
We have thus shown:

\begin{theorem}
  \label{thm:CH_main}
  Using $O(n^2)$ space and time, we can preprocess a set $L$ of $n$ lines 
  in the plane (given in general position), such that given any point set 
  $P$ with each point lying on a distinct line in $L$, 
  we can construct $\conv(P)$ in expected time $O(n\alpha(n)\log^* n)$.
  The same result holds if $L$ is a set of $n$ line segments whose
  supporting lines are in general position.
\end{theorem}

\textbf{Remark.}
An inspection of the proof of Theorem~\ref{thm:CH_main} shows
that the total expected conflict size, 
over all iterations $k$,  
is only $O(n)$. 

\subsection{Better Bounds for Oblivious Points}
\label{sec:oblivious}

We now present an improved solution under the \emph{obliviousness model},
where we assume that the points are oblivious to the random choices 
during the preprocessing step.
Specifically, this implies that an adversary cannot pick the point set $P$
in a malicious manner, as it is not aware of the random choices at the 
preprocessing step. 
This fairly standard assumption has appeared in 
various studies (see, e.g.,~\cite{AfshaniBaCh09,Chan09}).  
In the discussion at the end of this section we describe this issue in more detail.

\noindent\textbf{Preprocessing.}
We now construct a gradation 
$L_1 \subseteq L_2 \subseteq \cdots \subseteq L_{1+\log\log n} = L$ of 
the lines during the preprocessing phase, 
where the set sizes decrease geometrically.
Specifically, $|L_1| = y_1 = \lceil n/\log{n} \rceil$, 
and for $k = 2, \ldots, 1+\log\log{n}$, $L_{k-1}$ is a random 
subset of $L_k$ of size 
\begin{equation}
  \label{eq:geom_grad}
  |L_{k-1}| = y_{k-1} := \left\lceil{\frac{y_{k}}{2}}\right\rceil = 
  \left\lceil{\frac{1}{2} \cdot \frac{n}{2^{\log\log{n}-k+1}}}\right\rceil .
\end{equation}
We construct each arrangement $\A(L_k)$ in  $O(n^2 2^{2k-2}/\log^2{n})$
time, for a total $O(n^2)$ time 
over all gradation steps. 

\noindent\textbf{Query.}
Given an exact input $P$, we first follow the gradation produced 
at the preprocessing stage, and generate the corresponding gradation 
$P_1 \subseteq P_2 \subseteq \ldots \subseteq P_{1+\log\log{n}} = P$,
where $P_k = P \cap L_k$,
for all $k$. 
By the obliviousness assumption, each 
$P_{k-1}$ is an unbiased sample of $P_k$, 
so Lemma~\ref{lem:conf_lists} applies 
(see once again the discussion below). 
Moreover, the key observation is that in order to obtain 
$\conv(P_{k})$ from $\conv(P_{k-1})$, it suffices
to confine the search to the arrangement $\A(L_{k})$ 
instead of the entire 
arrangement $\A(L)$ as in Section~\ref{sec:quickCH}. 
Thus, we first construct $\conv(P_1)$ in $O(n)$ time as before.  
Next, to obtain $\conv(P_k)$ from $\conv(P_{k-1})$, we construct the zones of 
$\UH(P_{k-1})$ and $\LH(P_{k-1})$ in $\A(L_{k})$ in 
$O(y_{k}\alpha(y_{k}))$ time, and then compute the refined conflict lists 
just as in Section~\ref{sec:quickCH}. 
The overall 
expected time to produce these lists is $O(y_{k})$, 
totaling $O(n)$ 
over all steps.  
As in Section~\ref{sec:quickCH},
the expected time (of the final step) to compute $\UH(P_{k})$ 
and $\LH(P_{k})$ is 
$ 
O\left(y_{k} \log \left(y_{k}/y_{k-1}\right) \right) = 
O\left(y_{k}\right)$,
 
by~(\ref{eq:geom_grad}).
Thus, the expected running time at the $k$th step is dominated 
by the zone construction, so the overall expected running time 
is

$
O\left(\sum_{k=1}^{1+\log\log{n}} y_{k}\alpha(y_{k}) \right) = 
O\left(n \alpha(n)\right)$, 
as is easily verified. Thus, 

\begin{theorem}
  \label{thm:CH_obliviousness}
  Using $O(n^2)$ space and time, we can preprocess a set $L$ of $n$ 
  lines in the plane,
  such that for any point set $P$ with each point lying on a distinct line 
  of $L$, we can construct $\conv(P)$ in expected time $O(n \alpha(n))$ 
  assuming obliviousness.
\end{theorem}

\noindent\textbf{Discussion.}
The issue captured by the obliviousness assumption is: 
how much does the adversary know about the 
preprocessing phase?  If the adversary manages to obtain the coin flips 
performed during the 
preprocessing stage, then this enables a malicious choice of the input. 
This phenomenon is particularly striking in the case of \emph{hashing}: 
if the adversary knows the random choice of the hash function, 
a bad set of inputs can hash all keys to a single slot, 
completely destroying the hash table. On the other hand, if the
adversary is \emph{oblivious} to the hash function, the expected running
time per operation is only $O(1)$; see, e.g., 
\cite[Chapter~11]{CormenLeRiSt09}.

In our model we encounter a similar phenomenon. Even though the impact
is not as disastrous as for hashing, assuming obliviousness for the adversary
can improve our running time by a factor of $O(\log^*n)$. To illustrate 
the effect of obliviousness in our setting,
consider the scenario illustrated in Figure~\ref{fig:oblivious}. 
In this case, we have a set $L$ of lines and a random subset $L' \subseteq L$.
The adversary can pick the point set $P$ so that $P' := P \cap L'$ is a 
biased sample of $P$ in a sense that violates the properties of 
Lemma~\ref{lem:conf_lists}.  In particular, the total number of conflicts 
between the edges of $\conv(P')$ and $P$ may become quadratic, which makes 
the random incremental construction inefficient.  Nevertheless, if the 
adversary is oblivious with respect to the sample, the points in $P'$ 
behave as an unbiased sample of $P$.
\begin{figure}
  \begin{center}
    \includegraphics{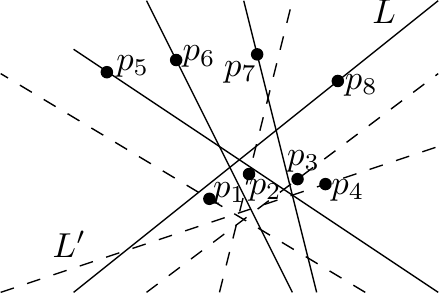} 
  \end{center}
  \caption{
      Illustrating the obliviousness assumption.
      The set $L$ contains all the lines in the figure.
      The set $L'$ of the dashed lines depicted in the figure is a random 
      subset of $L$.
      If the adversary knows this random subset $L'$, it can place the 
      points $p_1, \ldots, p_8$ as illustrated in the figure, thereby 
      constructing $\conv(p_1, \ldots, p_4)$ first.
      Then each edge on the upper hull $\UH(p_1, \ldots, p_4)$ is in 
      conflict with each of the remaining points $p_5, \ldots, p_8$.
  }
  \label{fig:oblivious}
\end{figure}

\subsection{Extensions and Variants}
\label{sec:extensions}

\noindent\textbf{Diameter- and width-queries.}
Given $\conv(P)$, we can easily compute the 
\emph{diameter} (i.e., a pair of points with maximum Euclidean distance) 
and the \emph{width} of $P$ 
(a strip of minimal width containing all the points in $P$) in linear 
time (see, e.g.,~\cite[Chapter 4]{PS-85}). Hence,
\begin{corol}
  \label{thm:diam_width}
  Using $O(n^2)$ space and time, we can preprocess a set $L$ of $n$ lines 
  in the plane, such that given a point set $P$ with each point lying on 
  a distinct line of $L$, the diameter or width of $P$ can be found in 
  expected time $O(n\alpha(n)\log^* n)$.  The expected running time becomes 
  $O(n \alpha(n))$ assuming obliviousness.
\end{corol}

\noindent\textbf{A trade-off between space and query time.}
Our data structure can be generalized
to support a trade-off between preprocessing time (and storage) and the 
query time, using a relatively standard grouping technique~\cite{AD-07,Chan96},
described as follows. 

\noindent\textbf{Preprocessing.}
Let $1 \le m \le n$ be a parameter, 
and, without loss of generality, assume that $n/m$ is an integer.
We partition $L$ into $m$ subsets $L_1, \ldots, L_{m}$ of 
size $n/m$ each, and construct the arrangements $\A(L_k)$, for 
$k=1, \ldots, m$, in overall time and storage $O(n^2/m)$ (cf.~\cite{AD-07,Chan96}).

\noindent\textbf{Query.}
Given an exact input $P$, we first construct $\conv(P_k)$,
where $P_k$ is the subset of points on the lines in $L_k$, 
$k=1, \ldots, m$, in $O((n/m) \alpha(n/m) \log^*{(n/m)})$ (assuming
obliviousness it is $O((n/m) \alpha(n/m))$) expected time, 
for a total expected time of $O(n \alpha(n/m) \log^*{(n/m)})$ 
(resp., $O(n \alpha(n/m))$) 
over all these subsets.
Having $\conv(P_k)$ at hand for all $k$, we merge $\UH(P_1)$, $\ldots$, $\UH(P_m)$ in 
$O(n \log{m})$ time~\cite{CormenLeRiSt09}, 
thereby producing a list $Q$ of points sorted according to their $x$-order. 
We then use Graham's scheme to
construct the upper hull of $Q$ (and thus of $P$) in $O(|Q|)$ time. We produce 
the lower hull of $P$ in an analogous manner. We have thus shown:

\begin{corol}
  \label{col:trade_off}
  Fix $1 \leq m \leq n$. 
  In total $O(n^2/m)$ time and space, we can preprocess a set
  $L$ of $n$ lines in the plane, such that given a point set $P$
  with each point lying on a distinct line of $L$, 
  we can construct $\conv(P)$ in expected time $O(n(\log{m} + \alpha(n/m)\log^*{(n/m)}))$.
  The running time becomes $O(n(\log{m} + \alpha(n/m)))$
  assuming obliviousness.
\end{corol}

Note that for small values of $m$, Corollary~\ref{col:trade_off} in fact
yields an improvement over Theorem~\ref{thm:CH_main}. Specifically,
by setting $m \eqdef 2^{\alpha(n)}$, we have that the space and preprocessing
requirement in Theorem~\ref{thm:CH_main} can be lowered to
$O(n^2/2^{\alpha(n)})$, while the expected query time remains
$O(n \alpha (n) \log^*(n))$.

\noindent\textbf{Discussion.}
As noted in the introduction, the bounds in Corollary~\ref{col:trade_off} are somewhat disappointing.
However, the study by Ali Abam and de Berg~\cite{AD-07} might provide (albeit, weak) evidence  
that these bounds are unlikely to be improved. Indeed, they have studied the \emph{kinetic sorting problem},
where we are given a set of $n$ points moving continuously on the real line, and the goal is to maintain
a structure on them so that at any given time the points can be sorted efficiently. 
Ali Abam and de Berg~\cite{AD-07} showed that even when the trajectories of these points
are just linear functions, then under the 
comparison graph model (see~\cite{AD-07} for the definition) one cannot 
answer a query faster than $cn\log{m}$ time using less than $c'n^2/m$ 
preprocessing time and storage, for appropriate absolute constants $c, c' > 0$.
As discussed in~\cite{AD-07}, this may indicate that better trade-offs for the 
kinetic convex hull problem seem unlikely. Nevertheless, it may not provide a rigorous proof, as
the analysis for the kinetic sorting problem strongly relies on the one-dimensionality of the points, 
and does not work for points in the plane, at least under the context of the proofs given in~\cite{AD-07}.
Still, we have chosen to present those details in this paper, as we tend to believe that bounds of this
kind could also apply to our problem (which is even more difficult than the kinetic convex hull problem,
as described in the introduction), and that a rigorous analysis could stem from the approach 
in~\cite{AD-07}.
This would imply that the trade-off bounds given in Corollary~\ref{col:trade_off} are nearly optimal.

\noindent\textbf{An output-sensitive algorithm.}
Our algorithm can be made sensitive to the size $h$
of the convex hull by adapting a technique of Ali Abam and de Berg~\cite{AD-07} that uses
\emph{gift wrapping queries}.
The setting for queries of this kind is as follows.
Let $Q$ be a point set, given an arbitrary point $p$
(not necessarily from $Q$) and a line $\ell$ through $p$, such 
that all points of $Q$ lie on the same side of $\ell$, report a point 
$q \in Q$ that is hit first when $\ell$ is rotated around $p$
(say, in clockwise direction).
 
\noindent\textbf{A search on the value of $h$.}
Since the output size $h$ is not given in advance,
we perform a search on its actual value, over at most $\log^*{n} - 1$ iterations, 
in the query step, and apply all tested values at the preprocessing step, as described below.
The tested values of $h$ are chosen in the following manner.
Let $h_i$ be the value of $h$ at the $i$th round.
Initially, $h_1 = 1$,
and put $h_i = 2^{(i-1)}$, for $i \ge 2$, where $2^{(\cdot)}$ is the power-tower function.
We continue the search as long as $h_i \le \log{n}$;
let $t$ be the number of rounds thus obtained.
By construction $t \le \log^*{n} - 1$.
When $h_i > \log{n}$, 
we stop the search and resort to the bound in Theorem~\ref{thm:CH_main}---see below.

\noindent\textbf{Preprocessing.}
At each round $i=1, \ldots, t$, 
we set a parameter $m_i$ to be
\[
m_i \eqdef \max \left\{1, \frac{n}{h_i\log{h_i}} \right\} ,
\]
and partition $L$ into $m_i$ roughly equal subsets $L^{(i)}_1, \ldots, L^{(i)}_{m_i}$.
We then proceed in a similar manner as described earlier for the trade-off between
space and query time.
That is, for each $k = 1, \ldots, m_i$, we construct the arrangement 
$\A(L^{(i)}_k)$ in overall time and storage $O(n^2/m_i)$.

The total time and storage consumed over all rounds $i$ is thus
\[
O\left(n^2 + \sum_{i=1}^{t} n h_i \log{h_i}\right) = 
O(n^2) ,
\] 
since the sum over the rounds $i$ is dominated by the last term,
which is $O(n \log{n}\log\log{n})$.

\noindent\textbf{Query.} 
Given a point set $P$ with each point lying on a distinct line of $L$, 
we construct $\conv(P)$ in an output-sensitive manner, as follows. 

At the $i$th round, let $P^{(i)}_k \eqdef P \cap L^{(i)}_k$, for
$k = 1, \ldots, m_i$.
Construct $\conv(P^{(i)}_1)$, $\conv(P^{(i)}_2)$, $\ldots$, $\conv(P^{(i)}_{m_i})$, as
in Section~\ref{sec:quickCH}.
This takes total time 
$O(n \alpha(n/m_i) \log^*{(n/m_i)}) = O(n \alpha(h_i)\log^*(h_i))$ (resp.
$O(n \alpha(n/m_i)) = O(n \alpha(h_i))$ assuming obliviousness).

The primitive operation we would like to obtain is a gift wrapping query on $P$.
To this end, we perform standard gift wrapping queries for each subset 
$P^{(i)}_k$ in $O(\log{(n/m_i)})$ time (see, e.g.,~\cite{PS-85}). This 
yields a set of $m_i$ candidates, from which we produce the final answer to the query. 
In total, a gift wrapping query takes $O(m_i\log{(n/m_i)}) = O(n/h_i)$ steps.

We now attempt to construct $\conv(P)$. 
We begin with a gift wrapping query for the 
leftmost vertex $p$ of $P$ and the vertical line $\ell_{p}$ passing 
through $p$. This yields a pair $(p',\ell')$, where $p'$ is the first point 
hit by $\ell$, and $\ell'$ is the line through $p$ and $p'$. We continue until 
(i) we hit $p$ again, or 
(ii) we have performed $h_i$ gift wrapping queries.
This results in a running time of 
\[
O\left(h_i \cdot m_i\log{(n/m_i)}\right) = O\left(h_i \cdot \frac{n}{h_i}\right) = O(n) .
\]
The round \emph{succeeds} if we reach $p$. Otherwise it fails, and
we proceed to round $i+1$. After $t \le \log^*n -1$ unsuccessful
rounds (i.e., if $h_i > \log{n}$), we compute $\conv(P)$ directly via 
Theorem~\ref{thm:CH_main}.

It is easy to verify that the actual number of rounds that we need
is at most $O(\log^* h)$.
Combining the bounds above, it follows that 
the overall query time is $O(n\alpha(h) (\log^* {h})^2)$
(or $O(n\alpha(h)\log^* {h})$ assuming obliviousness), as asserted.
We have thus shown:

\begin{corol}
  \label{thm:output_sensitive}
  In total $O(n^2)$ time and space, we can preprocess a set
  $L$ of $n$ lines in the plane, such that given a point set $P$
  with each point lying on a distinct line of $L$, 
  $\conv(P)$ can be found in expected time 
  $O(n\alpha(h)(\log^*{h})^2)$, where $h$ is the 
  output size. The expected running time becomes
  $O(n\alpha(h)\log^* h)$ assuming obliviousness.
\end{corol}

\section{Levels in Arrangements}
\label{sec:levels}

\noindent\textbf{Preliminaries.}
Let $L$ be a set of $n$ lines in the plane (in general position). 
Given a point $p$, the \emph{level} of $p$ with respect to $L$ is the 
number of lines in $L$ intersected by the open downward vertical ray 
emanating from $p$.
For an integer $k \ge 0$, the \emph{$k$-level} of the arrangement
$\A(L)$, denoted by $\lev_{k}(L)$,
is the closure of all edges of $\A(L)$ whose interior points have level $k$ 
with respect to $L$. It is a monotone piecewise-linear chain. In particular,
$\lev_0(L)$ is the so-called ``lower envelope'' of $L$;
see, e.g.,~\cite[Chapter~5.4]{SA-95} and Figure~\ref{fig:levels}. 
The \emph{$(\leq k)$-level} of $\A(L)$, denoted 
by $\lev_{\leq k}(L)$, is the complex induced by all cells
of $\A(L)$ lying on or below the $k$-level, and thus its edge set is
the union of $\lev_i(L)$ for $i=0, \ldots, k$; its overall 
combinatorial complexity is $O(nk)$ (see, e.g.,~\cite{CS-89, SA-95}).

In what follows we denote by $V_{q}(M)$ (resp., $V_{\leq q}(M)$) the set of
vertices of $\lev_{q}(M)$ (resp., $\lev_{\leq q}(M)$), where 
$q \geq 0$ is an integer parameter and $M$ is a set of lines in the plane.
It is easy to verify that the combinatorial complexity of $\lev_{q}(M)$
is at most $O(1 + |V_{q}(M)|)$ (see once again~\cite{SA-95}).
Throughout this section, we use the \emph{Vinogradov}-notation:
$f \ll g$ means $f = O(g)$ and $f \gg g$ means $f = \Omega(g)$. 
In addition, we write $\EX_{X}[\cdot]$ to emphasize that 
we take the expectation with respect
to the random choice of $X$ (the other variables are considered
constant).

\begin{figure}
  \begin{center}
    {\includegraphics[scale=0.4]{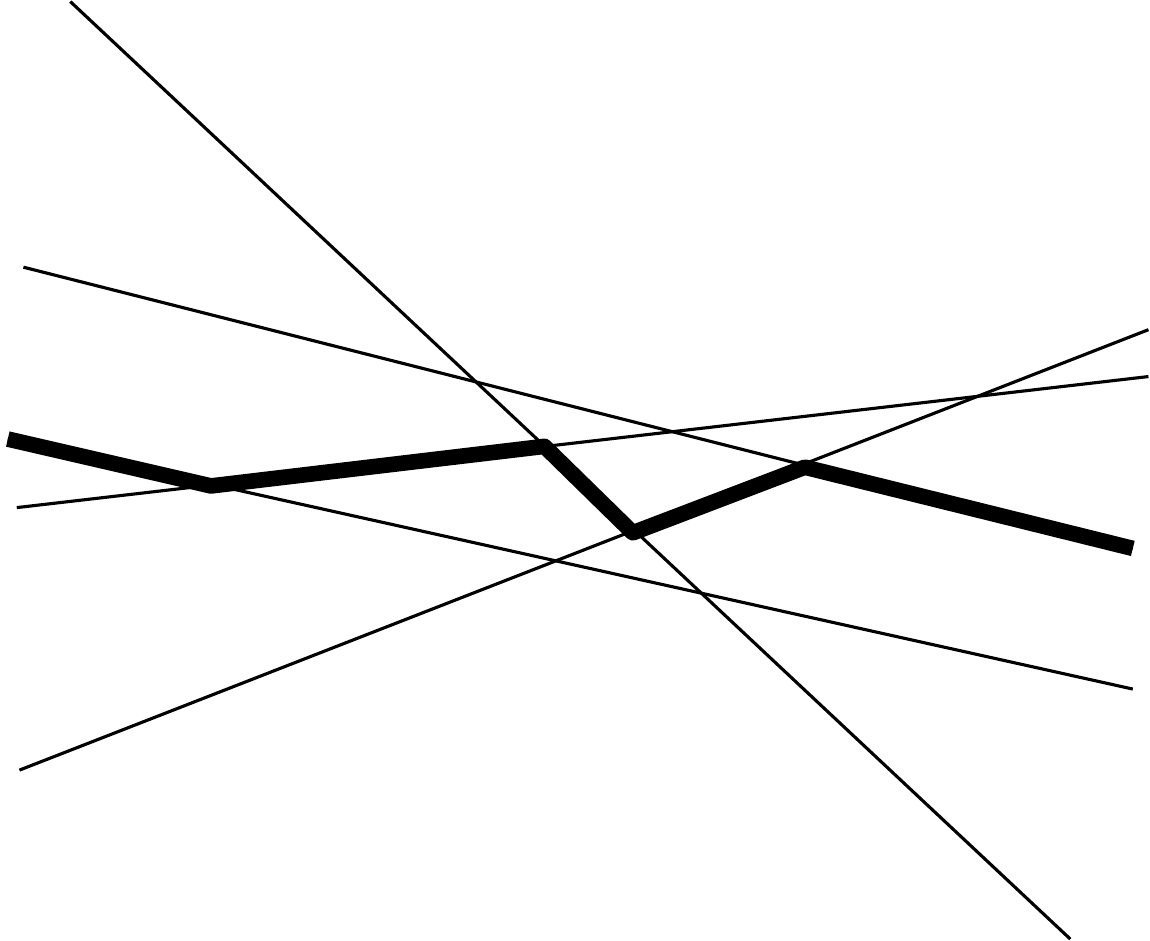} } 
  \end{center}
  \caption{
    The $2$-level in an arrangement of lines.
  }
  \label{fig:levels}
\end{figure}

The best currently known bound for the worst-case complexity 
of $\lev_{k}(L)$ is $O(n k^{1/3})$~\cite{Dey98}.
Nevertheless, since the overall combinatorial complexity of, say, 
$\lev_{\leq 2k}(L)$ is only $O(nk)$~\cite{CS-89}, 
it follows that the \emph{average} 
size of $V_{i}(L)$, for each $i=k, \ldots, 2k$, is only $O(n)$.
Specifically, we have (see also~\cite{EverettRovKr96} for a similar
property): 

\begin{claim}
  \label{clm:lin_avg}
  Let $\tk$ be a random integer in the range $\{k$, $\ldots$, $2k\}$. 
  Then, for any subset $S \subseteq L$, we have
  \[ 
  \EX_{\tk}[|V_{\tk}(S)|] \ll |S| .
  \]
\end{claim}

\begin{proof}
  The claim follows from the observation that 
  the total size of $V_{\leq 2k}(S)$ is 
  $O(|S|k)$, and each vertex appears in
  exactly two consecutive levels of $\A(S)$. 
\end{proof}

\noindent\textbf{The problem.}
In the sequel we study the following problem.
We are given a set $P = \{p_1, \ldots, p_n\}$ of $n$ points in the plane
(in general position), 
and we would like to compute a data structure such that, 
given any set $L = \{\ell_1, \ldots, \ell_n\}$ of $n$ lines satisfying 
$p_i \in \ell_i$, for $i=1, \ldots, n$, and  any parameter $k \ge 0$,
we can efficiently construct $\lev_{\leq k}(L)$.
This is a natural generalization of the problem studied in  
Section~\ref{sec:quickCH}.
Indeed, let us apply the standard duality transformation, where
a line $l :\; y = ax + b$ is mapped to the point $l^{*} = (a,-b)$,
and a point $p = (c,d)$ is mapped to the line $p^{*} :\; y = cx -d$
(see, e.g.,~\cite[Chapter~8]{deberg2008cga}).
Then $\lev_0(L)$ in the ``primal'' plane is mapped to the 
(upper) convex  hull of the points $L^{*}$ in the ``dual'' plane.
Everett~\etal~\cite{EverettRovKr96} showed that
$\lev_{\leq k}(L)$ can be constructed in $O(n \log n + nk)$ time,
and that this time bound is worst-case optimal (see also~\cite{Chan00}).
We show:
\begin{theorem}
  \label{thm:level_main}
  Using $O(n^2)$ space and time, we can preprocess a set $P$ of $n$ points 
  in the plane, such that given a set $L$ of lines with each line incident to 
  a distinct point of $P$,
  $\lev_{\leq k}(L)$ can be computed in expected time 
  $O(n\alpha(n)(\log^* n - \log^* k) + nk)$.
  The expected running time becomes $O(n\alpha(n) + nk)$ assuming
  obliviousness.
\end{theorem}

Theorem~\ref{thm:level_main} improves the ``standard'' 
bound of $O(n \log n + nk)$ for any $k = o(\log{n})$.
We combine ideas from Chan's algorithm for constructing
$(\leq k)$-levels in arrangements of planes in $\R^3$~\cite{Chan00} 
with the technique of Everett~\etal~\cite{EverettRovKr96}.
The preprocessing phase is fairly simple, but the details
of the query processing and its analysis are more intricate.
We begin with an overview of the approach, and then describe 
the query step and its analysis in more detail.

\noindent\textbf{An overview of the algorithm.}
The main ingredients of the algorithm are as follows.

\noindent\textbf{Preprocessing.}
Compute the arrangement $\A(P^{*})$ of the 
lines dual to the points in $P$ (and produce its vertical
decomposition) in $O(n^2)$ time and storage.

\noindent\textbf{Query.} 
We are given a set of lines $L$ as above, 
and an integer $k \ge 0$.
If $k \ge \log{n}$ we use the algorithm of Everett~\etal~\cite{EverettRovKr96}
to report $\lev_{\leq k}(L)$ in $O(n \log n + nk) = O(nk)$ time.
Otherwise, we compute a gradation  
$L_1 \subseteq L_2 \subseteq \cdots \subseteq L_{\log^* n - \log^* k + 1} \subseteq L$ 
of $L$. The sizes of the subsets $L_i$ are similar to those presented in 
Section~\ref{sec:quickCH} for the dual plane, but 
as soon as the number of lines in a subset of the gradation exceeds 
$\lceil{n/k}\rceil$, 
we complete the sequence in a single step by choosing the next subset 
to be the entire set $L$. 
As in Section~\ref{sec:quickCH}, we set $|L_1| := \lceil{n/\log n}\rceil$. 

We choose a random integer $\tk \in \{k, \ldots, 2k\}$. 
Then, at the first iteration, we construct $\lev_{\leq \tk}(L_1)$ in 
$O(nk)$ time, using the algorithm in~\cite{EverettRovKr96}.
At each of the following iterations $i$, we construct $\lev_{\leq \tk}(L_{i})$
from $\lev_{\leq \tk}(L_{i-1})$ 
(at the final step, we construct $\lev_{\leq \tk}(L)$ from 
$\lev_{\leq \tk}(L_{\log^* n - \log^* k + 1})$).
As observed above,
the random choice of $\tk$ guarantees that the expected 
complexity of each $\lev_{\tk}(L_{i})$ is only linear in $|L_i|$,
\footnote{
  We use the same value of $\tk$ throughout the entire process, 
  since the expected complexity of the $\tk$-level remains linear 
  in each iteration $i$. By linearity of expectation, the overall 
  expected size of the various $\tk$-levels is 
  linear in $\sum_{i=2}^{\log^* n - \log^* k + 1} |L_i|$.}
for each $i=2, \ldots, \log^* n - \log^* k + 1$, 
which is crucial for the analysis.
Finally, we eliminate from $\lev_{\leq \tk}(L)$ all portions lying above 
the (actual) $k$-level, in order to obtain the final 
structure $\lev_{\leq k}(L)$. 

To construct $\lev_{\leq\tk}(L_i)$ from $\lev_{\leq\tk}(L_{i-1})$,
we would like to proceed as follows.
We compute $\UH(\lev_{\leq\tk }(L_{i-1}))$
and subdivide it into semi-unbounded (in the negative $y$-direction) 
trapezoidal cells.
The first goal is to find for each such cell $\Delta$ the set of lines 
$C_\Delta \subseteq L_i$ which are in \emph{conflict} with 
$\Delta$ (that is, $\Delta \cap \ell \neq \emptyset$, for each $\ell \in C_\Delta$).
This goal is achieved by mapping $\UH(\lev_{\leq\tk }(L_{i-1}))$ to 
the dual plane and walking along its zone in $\A(P^{*})$.
The dual of $\UH(\lev_{\leq\tk }(S))$ is
a concave chain $\gamma$ (the lower envelope 
of the lines dual to the vertices of $\UH(\lev_{\leq\tk }(S))$), 
where each vertex $v$ of $\UH(\lev_{\leq\tk}(S))$ is mapped to an 
edge $v^{*}$ of 
$\gamma$ and each edge $e$ is mapped to a vertex $e^{*}$ of $\gamma$. 
Moreover, 
a line $\ell \in L$ below a vertex $v$ of $\UH(\lev_{\leq\tk}(S))$
is mapped to a point $\ell^{*}$ (on some line of $P^{*}$)  
above the corresponding edge $v^{*}$ of $\gamma$. As is easily verified, 
such a line $\ell$ intersects $\UH(\lev_{\leq\tk }(S))$. 
Otherwise, if $\ell$ lies above all the vertices 
of $\UH(\lev_{\leq\tk}(S))$, then 
$\ell \cap \UH(\lev_{\leq \tk}(S)) = \emptyset$, and this implies that 
$\ell^{*}$ lies below $\gamma$ in the dual plane. 
See Figure~\ref{fig:lines_chain}(a)--(b).

Having the lists $C_\Delta$ at hand, we construct for each $\Delta$ 
the structure
$\lev_{\leq \tk}(L_i)$ clipped to $\Delta$ by
(i) constructing $\lev_{\tk}(C_\Delta)$ (clipped to $\Delta$); 
(ii) clipping each line $\ell \in C_\Delta$ to its portion that lies below
$\UH(\lev_{\tk }(C_\Delta) \cap \Delta)$; (iii) constructing the 
arrangement of these portions within $\Delta$ 
(as observed in \cite{EverettRovKr96}, 
the actual level of these portions in $\A(L_i)$ does not exceed 
$2 \tk - 1$); and (iv) eliminating from the arrangement just computed
all portions lying above $\lev_{\tk}(C_\Delta) \cap \Delta$.
Finally, we glue the resulting structures together and report 
$\lev_{\leq \tk}(L_i)$. 

However, it would be too expensive to process each conflict list
$C_\Delta$ individually. Therefore, a crucial ingredient of the algorithm 
is to consider \emph{blocks} instead of just individual cells. 
Specifically, we gather contiguous cells into blocks and process
them all together. 
This partition is the key to reducing the number of cells considered 
in the update step; see Figure~\ref{fig:conf_list_block}.  The bulk
of the analysis lies in a careful balancing between the block sizes and their overall number,
and in particular showing that blocks with large conflict lists are scarce.

\subsection{Query Processing}

We now describe the query process and its analysis in more detail.
We first follow a gradation as described in the overview, and then
proceed to the update step.

\noindent\textbf{The update step.}
From now on we fix an iteration $i > 1$, and, with a slight
abuse of notation, put $S := L_{i-1}$ and $L := L_{i}$. 
Let $p := |S|/|L|$.
By definition, $S$ is a random sample of $L$ of size 
$\lceil{n/\log^{(i-1)} n}\rceil = p|L|$.
Given $\lev_{\leq\tk}(S)$, we first construct $\UH(\lev_{\leq\tk }(S))$. 

\begin{claim}
  \label{clm:UH_linear_time}
  The overall expected time to construct 
  $\UH(\lev_{\leq\tk }(S))$ is $O(|S|)$.
\end{claim}

\begin{proof}
  Using easy manipulations on the DCEL representing $\lev_{\leq\tk}(S)$,
  we can first locate a vertex $v$ of $\lev_{\tk}(S)$, and then proceed to its neighboring
  topmost vertex (say, to its left) by walking along its corresponding adjacent edge.
  We then continue progressing in this manner to the left.
  The vertices to the right of $v$ are explored analogously.
  Thus we can extract the sequence of vertices (and edges) along 
  $\lev_{\tk}(S)$, ordered from left to right.
  By Claim~\ref{clm:lin_avg}, its expected size is $O(|S|)$.
  Then we use Graham's scan on the resulting set of vertices.
\end{proof}

Next, we shoot vertical rays from each vertex of the hull 
$\UH(\lev_{\leq\tk }(S))$ 
in the negative $y$-direction. This gives a collection $\T_{\tk,S}$ of 
semi-unbounded trapezoidal cells covering $\UH(\lev_{\leq\tk}(S))$, 
and hence also $\lev_{\leq\tk}(L)$,
as is easily verified (see, e.g.,~\cite{Mat-92} for similar arguments). 
We group the cells in $\T_{\tk,S}$ into $O(|\T_{\tk,S}|(p/\tk))$ 
semi-unbounded vertical strips, each of which consists of 
$\tk/p$ contiguous cells. Such a vertical strip is called a
\emph{block}. 
Every block is bounded  by a convex chain from above, and by two 
vertical walls, one to its left and the other to its right. 
Let $\mathcal{B}$ be the set of all blocks.

We say that a line $\ell \in L$ is \emph{in conflict} with 
a cell $\Delta \in \T_{\tk,S}$, if $\Delta \cap \ell \neq \emptyset$. 
The conflict list $C_\Delta$ 
is then the set of all lines $\ell \in L$ in conflict 
with $\Delta$, and we put $c_\Delta := |C_\Delta|$. We similarly
define conflict lists $C_B$ and conflict sizes $c_B$ for each block
$B \in \mathcal{B}$.
Our next goal is to 
determine the conflict lists $C_B$ for each block. 

\begin{lemma}
  \label{lem:findconflicts}
  We can construct the conflict lists $C_B$, $B \in \mathcal{B}$,
  in overall time 
  \[ 
    O\left(n\alpha(n) + |V_{\tk}(S)| + \sum_{B \in \mathcal{B}} c_B\right).
  \]
\end{lemma}

\begin{proof}
First, we determine for each line $\ell \in L$ one
trapezoid $\Delta_\ell \in \T_{\tk,S}$ such that
$\ell \in C_{\Delta_\ell}$, if such a $\Delta_\ell$ exists.
This is done by a walk in the dual plane, as described
in the overview above and illustrated in Figure~\ref{fig:lines_chain}.
Specifically, we dualize $\UH(\lev_{\leq\tk}(S))$ to a 
concave chain $\gamma$.
\begin{figure}
  \begin{center}
    \includegraphics{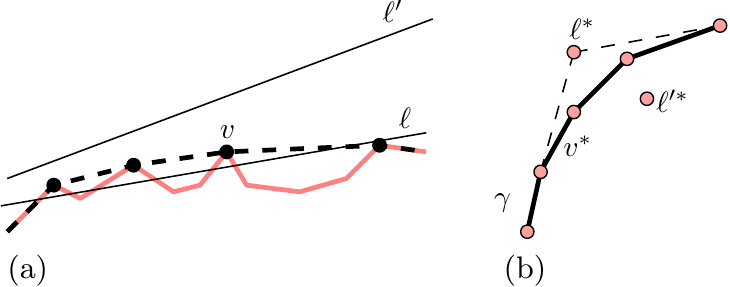}
  \end{center}
  \caption{
      (a) The $\tk$-level of $\A(S)$ is depicted by the
      lightly-shaded polygonal line in the figure, and
      $\UH(\lev_{\leq \tk}(S))$ is depicted by the dashed line.
      The line $\ell'$ passes above $\UH(\lev_{\leq \tk}(S))$, where
      $\ell$ passes below $v$ and thus meets $\UH(\lev_{\leq \tk}(S))$.
      (b) The dual scene of (a).
      The concave chain $\gamma$ is the dual of $\UH(\lev_{\leq \tk}(S))$.
      The line $\ell$ is mapped to the point $\ell^{*}$, where the pair
      of the dashed lines depict the visibility lines of $\ell^{*}$ to $\gamma$.
      The line $\ell'$ is mapped to the point $\ell'^{*}$ lying below $\gamma$.
  } 
  \label{fig:lines_chain}
\end{figure}
Using a similar technique as in Section~\ref{sec:quickCH}, we walk along the 
zone of $\gamma$ in $\A(P^{*})$ in order to determine, for each point $\ell^*$ 
corresponding to a line $\ell \in L$, its orientation with respect to $\gamma$.  
When $\ell^{*}$ lies above $\gamma$, we find an edge 
$v^{*}_{\ell^*}$ of $\gamma$ that is visible from $\ell^{*}$.
Using the corresponding vertex $v_\ell$ in the primal plane, we can determine
a cell $\Delta_\ell$ that is intersected by $\ell$.

Next, we determine for each such line $\ell$ a block $B$ that conflicts
with it, namely the block that contains $\Delta_\ell$. We then find all
blocks $B'$ with $\ell \in C_{B'}$ through a bidirectional walk from $B$. 
That is, we can determine if $\ell$ intersects the next block by checking whether
$\ell$ intersects any of its walls (otherwise, $\ell$ intersects its convex chain).
See Figure~\ref{fig:conf_list_block}.

The bound on the running time now follows using similar considerations
as in Section~\ref{sec:quickCH}.
\end{proof}

\begin{figure}
  \begin{center}
    \includegraphics{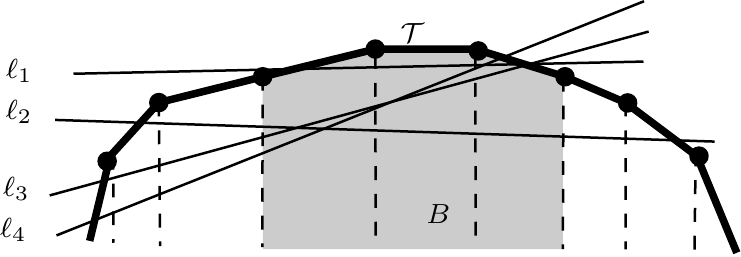}
  \end{center}
  \caption{
      The block $B$ is depicted by the lightly-shaded region.
      The line $\ell_1$ does not intersect any neighboring block, 
      whereas $\ell_3$ and $\ell_4$ also meet the left neighbor of $B$. 
      The line $\ell_2$ meets both neighbors.
  } 
  \label{fig:conf_list_block}
\end{figure}

Our next goal is to determine the $(\leq \tk)$-level clipped to $B$,
for each $B \in \mathcal{B}$.  To this end, we use a variant of the 
technique of Everett~\etal~\cite{EverettRovKr96}. 

\begin{lemma}
  \label{lem:per_region}
  Let $B \in \mathcal{B}$. The $(\leq \tk)$-level of $L$ clipped to $B$ can 
  be constructed in time $O(c_B\log{c_B} + (m_B + c_B) \log^2 k + a_B)$,
  where $c_B := |C_B|$, $m_B := |V_{\tk}(C_B) \cap B|$,
  and $a_B$ is the number of vertices of $\A(L)$ below 
  $\UH(\lev_{\tk}(C_R) \cap B)$.
\end{lemma}

\begin{proof}
  We apply the algorithm of Cole~\etal~\cite{ColeShYa87} in order to construct
  $\lev_{\tk}(C_B) \cap B$ in time $O(c_B\log{c_B} + (m_B  + c_B)\log^2 \tk)$.
  Note that this algorithm returns $\lev_{\tk}(C_B) \cap B$ as an $x$-monotone 
  polygonal chain $\zeta$ ordered from left to right.\footnote{The algorithm of 
    Cole~\etal~proceeds with a rotational sweep in the dual plane
    that keeps $\tk$ points of the input to the left of the sweep-line.
    In order to find only those vertices of the $\tk$-level which lie 
    inside $B$, we need to identify the appropriate initial orientation for 
    this line, but this is easily done by inspecting the intersection of 
    $C_B$ with the left boundary of $B$ and using a linear time selection  
    algorithm~\cite[Chapter~9]{CormenLeRiSt09}.}
  Next, we determine for each line $\ell \in C_B$ its first and last 
  intersections $w_1$, $w_2$ with $\zeta$ (if they exist). Clearly, the 
  portion of $\ell$ below $\UH(\zeta)$ is either (i) the line segment 
  $w_1 w_2$ (if both intersections exist); (ii) a ray with an endpoint 
  at $w_1$  (if $w_1$ is the only intersection with $\UH(\zeta)$) 
  or (iii) the full line $\ell$ clipped to $B$ (if it lies fully below $\zeta$).

  These intersections can easily be determined in $O(m_B)$ time by 
  walking along $\zeta$ and recording for each line $\ell$ the first and 
  last vertices of $\zeta$ that are incident to $\ell$ (if they exist); 
  at the representation of $\zeta$, 
  we also store the incident lines within each vertex.  
  A line that is not encountered during this process, 
  does not meet $\zeta$, and we can easily check whether it lies below $\zeta$.
  As observed above, each of these portions (clipped to $B$) is either the 
  (full) line $\ell$, a ray, 
  or a line segment. 
  Let $C_B'$ be the resulting set of these portions;
  by construction, $c_B' := |C_B'| \le c_B$. 
  Having this collection at hand, the computation of 
  $\lev_{\leq \tk}(C_B) \cap B$ 
  is almost straightforward. Indeed, we use an optimal line segment 
  intersection 
  algorithm~\cite{ChazelleEd92,CS-89} in order to
  compute the arrangement of $C_B'$ in time proportional to
  $c_B' \log c_B' + a_B \ll c_B \log c_B + a_B$,
  where $a_B$ is the number of intersections between the elements of $C_B'$.
  Note that some of these intersections may lie above the $\tk$-level, as they 
  are only guaranteed to be contained in $\UH(\zeta)$. Thus, at the final 
  step of the construction we eliminate such portions of the arrangement.
  This produces $\lev_{\leq \tk}(C_B) \cap B = \lev_{\leq \tk}(L) \cap B$.
  A key observation is the fact that all these 
  portions are actually contained in $\lev_{\leq (2\tk-1)}(C_B) \cap B$---see 
  below. 
\end{proof}

Finally, we glue all the resulting structures together and report 
$\lev_{\leq \tk}(L)$.  

\subsection{The Analysis}

We phrase our analysis below for a random subset $S$ of $L$
with $|S| = p|L|$, for some $p \in (0,1)$, as the value of $p$
varies at each iteration of the algorithm (as well as the final step).
We begin with the following key lemma that bounds the total size of
the large conflict sets. The proof is postponed
to \ref{app:levels}:
\begin{lemma}
  \label{lem:boundBv}
  Let $\tk$, $L$, $S$, $\T_{\tk,S}$, $C_\Delta$, $c_\Delta$, $p$ be defined as 
  above. Then, for any sufficiently large constant $\beta \geq 1$, we have
  \begin{equation}
    \label{eq:delta_conflist}
    \EX_{S}\Biggl[
      \sum_{\substack{\Delta \in \T_{\tk,S} \\ c_{\Delta} \ge \beta\tk/p^2 }} 
      c_{\Delta} (\log c_{\Delta} + \log^2 \tk + \log{(1/p)} )
      \Biggr]
    \leq  |L|e^{-\Theta(\beta)(\tk/p)}.
  \end{equation}
\end{lemma}

\noindent\textbf{Remarks.}
\noindent (1)
The bound in Lemma~\ref{lem:boundBv}  holds for any integer 
$\tk \in \{k\ldots 2k\}$. In particular, the analysis does not assume
a linear complexity bound on any of the levels of $\A(S)$ (and $\A(L)$);
see \ref{app:levels} for further details.
\newline\noindent (2)
It is easy to verify that the bound in Lemma~\ref{lem:boundBv} can be rewritten
when we apply the summation over all blocks.
That is,
\begin{equation}
  \label{eq:heavy_blocks}
  \EX_{S}\Biggl[ 
    \sum_{\substack{B \in \mathcal{B} \\ \Delta \in B , c_{\Delta} \ge \beta\tk/p^2 }} 
    c_{\Delta} (\log c_{\Delta} + \log^2 \tk + \log{(1/p)} ) \Biggr] 
  \leq  
  |L|e^{-\Theta(\beta)(\tk/p)} .
\end{equation}

\noindent\textbf{Bounding the expected running time.}
Adding the bounds in Lemmas~\ref{lem:findconflicts} and \ref{lem:per_region},
the running time to construct $\lev_{\leq \tk}(L)$ from
$\lev_{\leq \tk}(S)$ is asymptotically upper-bounded by
\begin{equation}\label{equ:runtime_bound}
n\alpha(n) + |V_{\tk}(S)| +
\sum_{B \in \mathcal{B}} c_B (\log c_B + \log^2 k) + m_B \log^2 k + a_B
\end{equation}
We bound each summand in turn. By Claim~\ref{clm:lin_avg}, we have
$\EX_{\tk}[|V_{\tk}(S)|] \ll |S| \leq |L|$.

\begin{claim}
  \label{clm:cRlogcR}
  We have:
  \[
  \EX_{S,\tk}\left[\sum_{B \in \mathcal{B}} c_B (\log{c_B}+\log^2\tk)\right] 
  \ll  
  |L|\left(\log(1/p) + \log^2 k\right).
  \]
\end{claim}
\begin{proof}
  We say that a line $\ell \in C_B$ is \emph{spanning}
  for a block $B \in \mathcal{B}$ if $\ell$ intersects both the left and the right walls of $B$,
  otherwise it is non-spanning. Note that every line can be non-spanning
  for at most two blocks.  Let $\beta \geq 1$ be a sufficiently large constant.
  We say that a block $B \in \mathcal{B}$ is 
  \emph{light} if it has at most $\beta \tk/p^2$ spanning lines and if 
  $c_B \leq \beta \tk^2/p^3$. Otherwise, $B$ is called \emph{heavy}.

  We split the summation, as follows:
  \begin{equation}\label{equ:split}
    \sum_{B \in \mathcal{B}} c_B \left(\log{c_B} + \log^2\tk\right) = 
    \sum_{\substack{B \in \mathcal{B} \\ 
	B \text{ is light}}} c_B \left(\log{c_B} + \log^2 \tk\right)
    +
    \sum_{\substack{B \in \mathcal{B} \\ B \text{ is heavy }}}
    c_B \left(\log c_B + \log^2\tk\right).
  \end{equation}
  Let us first consider the sum over the light blocks. Let  $B$ be a light block. 
  By definition, we have $\log c_B \ll \log(1/p)  + \log k$. Furthermore, write
  $c_B = c_B^{s} + c_B^{n}$, where $c_B^{s}$ is the number of spanning
  lines in $C_B$, and $c_B^{n}$ is the  number of non-spanning 
  lines in $C_B$. Observe that
  \[
\sum_{\substack{B \in \mathcal{B} \\ B \text{ is light}}}
    c_B^{s} \ll |\mathcal{B}|(\tk/p^2),
  \]
 since each light block can have only $O(\tk/p^2)$ spanning lines,
 and that
  \[
\sum_{\substack{B \in \mathcal{B} \\ B \text{ is light}}}
    c_B^{n} \ll |L|,
  \]
  since the total number of non-spanning lines is at most $2|L|$, over 
  all blocks in $\mathcal{B}$.
  It follows that
  \begin{align*}
    \sum_{\substack{B \in \mathcal{B} \\ B \text{ is light}}} c_B 
    \left(\log{c_B} + \log^2 \tk\right)
    &\ll 
    \left(\sum_{\substack{B \in \mathcal{B} \\ B \text{ is light}}}
    c_B^{s} + c_B^{n}\right)
    \left(\log(1/p) + \log^{2} \tk\right)\\
    &\ll
    \left(|\mathcal{B}|(\tk/p^2) + |L|\right)
    \left(\log(1/p) + \log^2 \tk\right).
  \end{align*}
  Now we bound $|\mathcal{B}|(\tk/p^2)$.
  Since each block contains $p/\tk$ contiguous cells, 
  we have $|\mathcal{B}| \ll |\T_{\tk,S}|(p/\tk)$.
  Furthermore, because the number of cells is bounded by the number
  of vertices on the $\tk$-level, we get
  \[
    \EX_{S,\tk}[|\T_{\tk,S}|] \ll \EX_{S,\tk}[|V_{\tk}(S)|] \ll
      \EX_{S}[|S|] = |L|p,
  \]
  using Claim~\ref{clm:lin_avg} and the definition of $p$ as $|S|/|L|$.
  Thus,
  \[
    \EX_{S,\tk}[|\mathcal{B}|(\tk/p^2)] \ll \EX_{S,\tk}[|\T_{\tk,S}|/p] 
    \ll |L|.
  \]
  Therefore,
  \[
  \EX_{S,\tk}\left[\sum_{\substack{B \in \mathcal{B} \\ B \text{ is light}}} 
    c_B (\log{c_B} + \log^2 \tk)\right]\\
  \ll 
  |L|
  \left(\log(1/p) + \log^{2} k\right).
  \]

  To bound the sum over the heavy blocks in (\ref{equ:split}), 
  observe that by definition a heavy block $B$
  must contain a cell $\Delta$ with $c_\Delta > \beta \tk/p^2$: either
  there are more than $\beta \tk/p^2$ spanning lines, in which case all
  the cells in $B$ have this property, or $c_B > \beta \tk^2/p^3$, in which
  case the claim follows from the fact that $B$ contains only $\tk/p$ cells.
  Let $\Delta^* \in B$ be the cell that maximizes $c_\Delta$ for 
  $\Delta \in B$. Clearly, we have $c_{\Delta^*} > \beta \tk/p^2$ and
  $c_B \leq (\tk/p)c_{\Delta^*}$.
  Hence,
  \begin{align*}
    \sum_{\substack{B \in \mathcal{B} \\ B \text{ is heavy }}}
    c_B \left(\log c_B + \log^2\tk\right) 
    & \ll
     \sum_{\substack{B \in \mathcal{B} \\ 
	\Delta \in B, c_{\Delta} \ge \beta\tk/p^2 }} 
    (\tk/p) \cdot c_{\Delta} \left(\log {(c_{\Delta} \tk/p )} + \log^2 \tk\right) \\
    &\ll
    (\tk/p) \sum_{\substack{B \in \mathcal{B} \\ 
	\Delta \in B, c_{\Delta} \ge \beta\tk/p^2 }} 
    c_{\Delta} \left(\log c_{\Delta} + \log^2 \tk + \log(1/p)\right).
  \end{align*}

  By~(\ref{eq:heavy_blocks}),
  the expectation (over $S$) of the latter sum
  is at most
  \[
    (\tk/p)|L|e^{-\Theta(\beta)(\tk/p)} \ll |L|,
  \]
  for $\beta$ sufficiently large.
\end{proof}

\noindent\textbf{Remark.}
In the analysis of Lemma~\ref{clm:cRlogcR} concerning the bound for the heavy blocks, 
each such block $B$ may consist of both
heavy cells (that is, cells $\Delta$ with $c_{\Delta} \ge \beta\tk/p^2$) and light ones.
At first glance, one may suspect that the overall contribution of the light cells should 
have the bound $O(|L|(\log(1/p) + \log^{2} k))$, as obtained in the case for light blocks.
Nevertheless, since these cells belong to a heavy block, the actual bound is smaller,
and in fact follows from the property that the number of heavy blocks is eventually much
smaller than the number of light blocks (this property is an easy consequence 
of Lemma~\ref{lem:boundBv}).

\begin{claim}
  \label{clm:mR}
  We have:
  $\sum_{B \in \mathcal{B}} m_B \log^2 \tk \ll |L|k$.
\end{claim}
\begin{proof}
  Recall that $m_B = |V_{\tk}(C_B) \cap B| = |V_{\tk}(L) \cap B|$. 
  Using Dey's bound on the size of the $\tk$-level~\cite{Dey98},
  it follows that
  $\sum_{B \in \mathcal{B}} m_B = |V_{\tk}(L)| \ll |L|k^{1/3}$.
  The claim is now immediate.
\end{proof}

\begin{claim}
  \label{clm:aR}
  We have:
  $\sum_{B \in \mathcal{B}} a_B \ll |L|k$.
\end{claim}
\begin{proof}
  Everett~\etal~\cite{EverettRovKr96} have shown that no element
  in $C_B'$ contains a point which lies above 
  $\lev_{2\tk-1}(C_B) \cap B$. Since all sets $C_B'$ are clipped to $B$, 
  for each $B \in \mathcal{B}$, it follows that all portions of the various 
  arrangements that we construct, over all $B \in \mathcal{B}$, lie within 
  $\lev_{\leq 2\tk-1}(L)$. Hence,
  \[
  \sum_{B \in \mathcal{B}} a_B \ll |L|\tk \ll |L|k.
  \] 
\end{proof}

We thus conclude:

\begin{corol}
  \label{cor:per_iteration}
  The total expected running time for the $i$th iteration is
  \[ 
  O\left(n\alpha(n) + |L_i|\left(k + \log{\left(\frac{|L_{i}|}{|L_{i-1}|}\right)}\right)\right) .
 \] 
\end{corol}
\begin{proof}
This follows by substituting the bounds from
Claims~\ref{clm:cRlogcR}--\ref{clm:aR} into
(\ref{equ:runtime_bound}), by using that $\log^2 k = O(k)$, and by remembering that
we set $L = L_i$, $S = L_{i-1}$ and $p = |S|/|L|$.
\end{proof}

Note that 
\[
\sum_{i=2}^{\log^* n - \log^* k + 1} |L_i|k =
\sum_{i=2}^{\log^* n - \log^* k + 1} nk/\log^{(i)} n
\ll nk,
\]
since the sequence $\bigl\{1/\log^{(i)} n\bigr\}_{i=\log^* n - \log^* k}^{2}$ 
decreases faster than any geometric sequence. 
Moreover, for all but the last iteration, we have
$|L_{i}|\log(|L_{i}|/|L_{i-1}|) \ll |L_{i}|\log^{(i)} n \ll n$. 
At the last iteration, we have $|S| \geq n/k$, so
$\log (|L|/|S|) \leq \log k$, and thus
\[
\sum_{i=2}^{\log^* n - \log^* k + 1} |L_i| \log^{(i)} n  + |L|\log{k} \\
\ll 
n(\log^* n - \log^* k + \log{k}) .
\]
It thus follows that the overall expected running time is
$O(n \alpha(n)(\log^* n - \log^*k) + nk)$.
It is easy to verify that when $k=0$ we obtain the same asymptotic time bound
as in Theorem~\ref{thm:CH_main}.

\noindent\textbf{A faster algorithm under the obliviousness assumption.}
Similar to Section~\ref{sec:oblivious},
the expected running time can be improved to
$O(n (\alpha(n) + k))$ assuming obliviousness.
As before, we now compute a gradation during the preprocessing
phase:
$P_1 \subseteq P_2 \subseteq \cdots \subseteq P_{1+ \log\log n} = P$
with $|P_1| \ll n/\log n$ and $|P_{i}| = 2|P_{i-1}|$, and we compute
each of the arrangements $\mathcal{A}(P_i^{*})$ in the dual plane. 

The algorithm for processing a set of lines $L$, with each line containing
exactly one point, is just as above, with two major differences:
first, we compute the gradation for $L$ by using the precomputed
gradation for $P$. Second, during the $i$th iteration
we use $A(P_i^{*})$ instead of $A(P^{*})$ to determine the zone of $\gamma$.
Using similar considerations as in Section~\ref{sec:oblivious}, 
the bound in Corollary~\ref{cor:per_iteration} now becomes 
$O(|L_i|(\alpha(n) + k))$, because $\log(|L_i|/|L_{i-1}|) = 1$. Summing 
over the various iterations $i$ and the final step yields the bound
$O(n(\alpha(n) + k))$, as asserted.
The total storage requirement remains $O(n^2)$. 

\smallskip\noindent
This at last concludes the proof of Theorem~\ref{thm:level_main}.

\section{Lower Bounds}
\label{sec:lowerbounds}

In this section we study problems where
preprocessing $\A(L)$ is unlikely to decrease the query time to $o(n\log{n})$ 
(at least under some computational models).

\noindent\textbf{Delaunay triangulations.}
It has already been observed in~\cite{BuchinLoMoMu11,vKreveldLoMi10}
that for some sets $L$, even when we have $\A(L)$ precomputed, there are point 
sets,  with each point lying on a distinct line, such that their Delaunay 
triangulation 
cannot be constructed in $o(n \log n)$ time
(albeit sometimes one can obtain better bounds if each point lies on
a fat region given in advance~\cite{BuchinLoMoMu11,LoefflerSn10}). 
This lower bound holds in the classic algebraic computation tree
model~\cite[Chapter~16]{AroraBa09}, and it 
essentially comes from a construction due to Djidjev and 
Lingas~\cite{DjidjevLi95}.
Specifically, 
they showed that when the points are sorted in just a single direction, 
one cannot compute their Delaunay triangulation in less than $\Omega(n \log n)$ time. 
Thus, if $L$ is a set of vertical lines, we can only anticipate the $x$-order 
of the points (received later), from which the lower bound follows.
Note that this lower bound also implies that no speedup is possible for
computing the Euclidean minimum spanning tree (EMST), since the
Delaunay triangulation can be constructed in linear time once
the EMST is known~\cite{ChinWa99,KleinLi96}.

\noindent\textbf{Closest Pairs.}
Finding the closest pair in a point set
is somewhat easier than the Delaunay triangulation problem (since the
latter has an edge between the closest pair~\cite{deberg2008cga}),
but is often harder than computing convex hulls
(except perhaps when the model of computation provides the floor
function as well as a source of randomness, see, e.g.,~\cite{KM-95}).
Formally, the problem is defined as follows: 
given a set $L = \{\ell_1, \ldots, \ell_n\}$ of lines in
the plane, compute a data structure such that given any 
point set $P = \{p_1, \ldots, p_n\}$ with $p_i \in \ell_i$
for $i = 1, \ldots, n$, we can quickly find a pair 
$(p_i,p_j) \in P \times P$ of distinct points that 
minimizes $\| p_i - p_j \|$. 
Incorporating the lower bound by Djidjev and Lingas~\cite{DjidjevLi95},
we show the following:
\begin{prop}
  There exists a set $L = \{\ell_1, \ldots, \ell_n\}$ of lines in the plane,
  such that for any point set $P$ with each point lying on a distinct line 
  of $L$, finding the closest pair in $P$ (after preprocessing $L$) 
  requires $\Omega(n \log n)$ operations under the algebraic computation 
  tree model.
\end{prop}

\begin{proof}
Consider the problem \textsc{Fuzzy-2-Separation}:
for a sequence $x_1, \ldots, x_n$ in $\R$, 
output \textsc{No}, if there exists a pair $1 \leq i < j \leq n$ with 
$|x_i - x_j| \leq 1$, and \textsc{Yes}, if for each pair
$1 \leq i < j \leq  n$ we have $|x_i - x_j| \geq 2$.
In all other cases the answer is arbitrary.

\begin{claim}
  \label{clm:fuzzy}
  Any algebraic decision tree for the problem \textsc{Fuzzy-2-Separation}
  has depth $\Omega(n \log n)$.
\end{claim}

\begin{proof}
This follows from a straightforward application of the technique of 
  Ben-Or~\cite{BenOr83}. The only somewhat non-standard feature is the
  need to deal with fuzziness. Let 
  \[
  W_1 = \{(x_1, \ldots, x_n) \in \R^n \mid \forall\, 1 \leq i < j \leq n:
  |x_i - x_j | > 1\},
  \]
  and let 
  \[
  W_2 = \{(x_1, \ldots, x_n) \in \R^n \mid \forall\, 1 \leq i < j \leq n:
  |x_i - x_j | \geq 2\}.
  \]
  
  Let $T$ be a decision tree for \textsc{Fuzzy-2-Separation},
  and let $W = T^{-1}(\textsc{Yes}) \subseteq \R^n$ be the set of
  inputs that lead to a leaf in $T$ labeled \textsc{Yes}. By definition,
  we have $W_1 \supseteq W \supseteq W_2$. It now follows that
  $W$ has at least $n!$ different connected components, since
  the $n!$ inputs $x_\pi = (2\pi(1), 2\pi(2), \ldots, 2\pi(n))$ for
  any permutation $\pi$ of $\{1, \ldots, n\}$ are all contained
  in $W_2$ and reside in different connected components of $W_1$
  (see~\cite[Theorem~16.20]{AroraBa09} for this standard technique).
  Hence, Ben-Or's result~\cite{BenOr83} implies that $T$ has
  depth $\Omega(\log n!) = \Omega(n \log n)$.
\end{proof}
The reduction from \textsc{Fuzzy-2-Separation} to
closest pair queries is almost straightforward.
For $i = 1, \ldots, n$, let $\ell_i$ be the horizontal line 
$\ell_i : y = i/n$, and let $L = \{\ell_1, \ldots, \ell_n\}$.
Thus, the only information we can precompute from $L$ is exactly this order.
Given an instance $(x_1, \ldots, x_n)$ of \textsc{Fuzzy-2-Separation}, 
we map each 
$x_i$ to a point $p_i = (x_i, i/n) \in \ell_i$,
and then find the closest pair in the resulting point set. 
If the distance of the closest pair is greater than $2$, 
our algorithm outputs \textsc{Yes}, otherwise it outputs \textsc{No}. 
Clearly, the overhead for this reduction is linear.
We are now left to show the correctness of the reduction.
Indeed, if $|x_i - x_j| \geq 2$, for every pair of indices 
$1 \leq i < j \leq n$, then clearly 
$\| p_i - p_j \| \geq \sqrt{4 + 1/n^2} > 2$, and this
in particular applies for the closest pair of points.
Otherwise, if there exists a pair $1 \leq i < j \leq n$ with $|x_i - x_j| < 1$,
then $\|p_i - p_j \| \leq \sqrt{1 + 1} < 2$ (and this also upper bounds 
the distance between the closest pair), so the reduction reports the 
correct answer on all mandatory 
\textsc{Yes} and \textsc{No} instances, as asserted.
\end{proof}

\noindent\textbf{Convex hull in three dimensions.} 
Returning to the convex hull problem, we next study its extension 
to three dimensions.
That is, given a set $H = \{h_1, h_2, \ldots, h_n\}$ of $n$ planes in $\R^3$,
we would like to compute a data structure, so that for any point set 
$P = \{p_1, \ldots, p_n\}$ with $p_i \in h_i$,
$i = 1, \ldots, n$, we can construct $\conv(P)$ quickly. 
Since the complexity of the convex hull in both 
$\R^2$ and $\R^3$ is only linear, and since there are several algorithms that
construct the convex hull (in both cases) in the same asymptotic running time 
(see, e.g.,~\cite{deberg2008cga, CS-89}),
one may ask if a three-dimensional convex hull query can be answered
in $o(n \log{n})$ time as well.
Using the well-known lifting transformation~\cite{SA-95}, one can quickly derive
a lower bound from the result about Delaunay triangulations
mentioned above, but below we also give simple direct reduction
(which follows immediately from a result of Seidel~\cite{Seidel84}).

\begin{prop}
  \label{prop:planes3D}
  There is a set $H = \{h_1, \ldots, h_n\}$ of planes in $\R^3$,
  such that for any point set $P$ with each point lying on a distinct plane of $H$,
  constructing $\conv(P)$ (after preprocessing $H$) 
  requires $\Omega(n \log n)$ operations under the algebraic computation 
  tree model.
\end{prop}

\begin{proof}
Let $h_i$ be the plane defined by the equation $z = i$, 
for $i = 1, \ldots, n$, and let $H = \{h_1, \ldots, h_n\}$.
We give a reduction from planar convex hulls to computing three-dimensional
convex hulls of point sets, where each plane in $H$ contains precisely one such point. 
Let $P = \{p_1, \ldots, p_n\}$ be a set of points in the plane,
and, for $i = 1, \ldots n$, let $\widehat p_i := (p_{ix}, p_{iy}, i)$, that is, 
the point obtained by lifting $p_i$ to $h_i$. 
As observed by Seidel~\cite[Section~IV]{Seidel84}, 
to compute the planar convex hull $\conv(P)$, 
it suffices to perform a convex hull query for 
$\widehat P = \{\widehat p_1, \widehat p_2, \ldots, \widehat p_n\}$ and
then project the result onto the $xy$-plane.
It is shown in~\cite{Seidel84} that once we have $\conv(\widehat P)$ at hand,
the time to project it onto the $xy$-plane (and then extract the actual 
planar convex hull) is only $O(n)$.
Thus the overhead of the reduction is linear, as is easily verified.  
The result now follows from the standard $\Omega(n \log n)$
lower bound for planar convex hulls in the algebraic computation tree model
(see, e.g.,~\cite{BenOr83}).
\end{proof}

\noindent\textbf{Sorting.}
Interestingly, a similar approach also shows that
\emph{sorting} requires $\Omega(n \log n)$ operations 
under the algebraic computation tree model. 
We have a set $L = \{\ell_1, \ldots, \ell_n\}$ 
of $n$ lines in the plane, and we wish to compute a data structure 
such that for any
set $P = \{p_1, \ldots, p_n\}$ of points with 
$p_i \in \ell_i$, $i=1, \ldots, n$, we can quickly
sort these points according to their $x$-order.

\begin{prop}
  \label{prop:sorting}
  There exists a set $L = \{\ell_1, \ldots, \ell_n\}$ of lines in the plane,
  such that for any point set $P$ with each point lying on a distinct line of $L$,
  sorting $P$ according to its $x$-order (after preprocessing $L$) 
  requires $\Omega(n \log n)$ operations under the algebraic computation 
  tree model.
\end{prop}

\begin{proof}
Let $X = \{x_1, \ldots, x_n\} \subseteq \R$. 
For $i = 1, \ldots, n$, let $\ell_i$ be the line 
$\ell_i : y = i$, and let $L = \{\ell_1, \ldots, \ell_n\}$.
We now lift each $x_i$ on $\ell_i$, and obtain the point $p_i := (x_i,i)$,
$i=1,\ldots, n$; let $P$ denote this set of points.
It is now easy to see that the $x$-order of $P$ yields the sorted order
for the numbers in $X$, and that this reduction has a linear running time.
\end{proof}

\section{Concluding remarks}
Note that Proposition~\ref{prop:sorting}, which has a 
straightforward proof, has an intriguing implication emphasizing a 
main contribution of this paper: while the ``standard'' planar convex hull 
and sorting problems are basically equivalent in terms of hardness 
(e.g.,~\cite{deberg2008cga}),  in our 
setting convex hull queries are in fact \emph{easier}. 
This improvement stems from the ``output-sensitive nature'' of convex hulls: 
points inside the hull are irrelevant to the computation, and the
information provided by $L$, combined with our update technique, 
allows us to quickly discard those non-extremal points, and not further 
process them in following iterations.
In our setting the two problems become equivalent if the input points are
in convex position.
Then, Proposition~\ref{prop:sorting} does 
not apply, since the points are sorted along two directions, 
and having the order according to one of them immediately implies 
the order according to the other. 

Our study raises several open problems.
The first one is whether the $\log^*{n}$ 
factor in the query time bound 
is indeed necessary for both convex hull and $(\le k)$-level queries. 
We conjecture it to be an artifact of the technique 
and that the
actual running times are $O(n \alpha(n))$ and $O(n(\alpha(n) + k))$ for the two
respective problems (as in the obliviousness model).
Another problem concerns the
case of convex hulls for points restricted  to 
three-dimensional \emph{lines}.
In this case, the lower bound in Section~\ref{sec:lowerbounds}
does not apply. 
Moreover, if the lines are \emph{parallel}, a simple variant of
our approach yields expected query time $O(n \log\log n)$ with polynomial
preprocessing and storage. 
Is there a better bound? What happens in the general case?

\paragraph*{Acknowledgments}
The authors wish to thank Maarten L\"offler for suggesting the problem and
for interesting discussions,
and Boris Aronov and Timothy Chan for helpful discussions.

We would like to thank the anonymous referees for their careful reading
of the paper and for numerous insightful comments that improved the
quality of the paper.

\bibliographystyle{abbrv}
\bibliography{lines}

\appendix

\section{Levels in Arrangements}
\label{app:levels}

\begin{proof}[Proof of Lemma~\ref{lem:boundBv}:]
We actually consider the sum over all $\Delta \in \T_{\tk,S}$ 
with $c_{\Delta} \ge 2\beta'\tk/p$,  where $\beta' > 1$ is a constant
to be fixed shortly. Then the lemma follows by 
choosing $\beta \ge 2\beta'$. 
In what follows, with a slight abuse of notation, we denote $\beta'$ by $\beta$.
For every vertex $v$ of $\A(L)$, let $C_v$ be the set of lines intersecting
the (open) downward vertical ray emanating from $v$, and put $c_v \eqdef |C_v|$.
Every vertex of $\lev_{\tk}(S)$ bounds at most two cells in 
$\T_{\tk,S}$, and for every $\Delta \in T_{\tk,S}$
any line in $C_\Delta$ passes under at least one vertex of $\Delta$.
Thus, we have $c_\Delta \leq 2\max\{c_{v_1},c_{v_2}\}$, where 
$v_1$, $v_2$ are the two vertices of $\Delta$.
We thus have:
\[
\sum_{\substack{\Delta \in \T_{\tk,S} \\ c_{\Delta} \ge 2\beta\tk/p^2 }} 
  c_{\Delta} (\log c_{\Delta} + \log^2 \tk + \log{(1/p)} )
\ll 
\sum_{\substack{v \in  V_{\tk}(S) \\ c_v \ge \beta\tk/p^2 }} 
  c_v (\log c_v + \log^2 \tk + \log{(1/p)}).
\]
Now, let $v$ be a vertex of $\A(L)$ at
level $c_v \geq (\tk/p^2)-1$, and let $\ell_1, \ell_2$ be the two lines
defining $v$. The vertex $v$ appears in $V_{\tk}(S)$ precisely if
(i) $\ell_1$ and $\ell_2$ are in $S$; and (ii) $S$ contains $\tk-1$ or $\tk$ lines
below $v$. Thus,
\begin{align*}
 \Pr[v \in V_{\tk}(S)] 
&= \Pr[\{\ell_1, \ell_2\} \subseteq S \wedge
  |S \cap C_v| \in \{\tk-1, \tk\}]  \label{equ:}\\
 &=\left(\binom{|L|-2}{p|L|-2}/\binom{|L|}{p|L|}\right)\cdot
 \,\Pr\left[|S \cap C_v| \in \{\tk-1, \tk\} 
 \mid \{\ell_1, \ell_2\} \subseteq S \right] \notag.
\end{align*}
Conditioned on $S$ containing $\{\ell_1, \ell_2\}$, the
sample $S' \eqdef S \setminus \{\ell_1, \ell_2\}$ is 
a random $(p|L|-2)$-sample from the set 
$L' \eqdef L \setminus \{\ell_1, \ell_2\}$.
Hence, $|S' \cap C_v|$ follows a 
hypergeometric distribution, so
Hoeffding's bound~\cite{Chvatal79,Hoeffding63} implies that
\begin{multline*}
  \Pr[|S \cap C_v| \in \{\tk-1, \tk\}\mid \{\ell_1,\ell_2\} \subseteq S]
  \leq \Pr[|S' \cap C_v|/|S'| \leq \tk/|S'|]\\
  = \Pr[|S' \cap \overline{C_v}|/|S'| \geq 1- \tk/|S'|]
  \leq
  \left(\left(\frac{1-c_v/|L'|}{1-\tk/|S'|}\right)^{1-\tk/|S'|}
  \left(\frac{c_v/|L'|}{\tk/|S'|}\right)^{\tk/|S'|}\right)^{|S'|},
\end{multline*}
recalling that $c_v = |C_v|$ denotes the number of lines below $v$.
Now note that 
\[
p|L'| = p(|L|-2) = |S| - 2p = |S'| + 2 -2p.
\]
Thus, writing $c_v = t \cdot \tk/p$, for some appropriate $t \geq \beta/p$, we get
\begin{multline*}
  \Pr[|S \cap C_v| \in \{\tk-1, \tk\} \mid \{\ell_1, \ell_2\} \subseteq S]\\
  \leq
  \Biggl(\Bigl(\frac{1-t\tk/(|S'|+2(1-p))}{1-\tk/|S'|}\Bigr)^{1-\tk/|S'|}\cdot 
  \Bigl(\frac{t\tk/(|S'|+2(1-p))}{\tk/|S'|}\Bigr)^{\tk/|S'|}\Biggr)^{|S'|}.
\end{multline*}
To simplify this, we first observe that
\[
1-\frac{t\tk}{|S'|+2(1-p)} \leq 1 - \frac{t\tk}{|S'| + 2}  \leq
1 - \frac{t\tk}{|S'| + |S'|/3} =
1 - \frac{3t\tk}{4|S'|},
\]
since we can assume $S'$ is large enough so that $2 \leq |S'|/3$. Therefore,
\[
\frac{1-t\tk/(|S'|+2(1-p))}{1-\tk/|S'|} \leq
\frac{1-(3t/4)\tk/|S'|}{1-\tk/|S'|} = 
1 - \frac{(3t/4 -1)\tk/|S'|}{1-\tk/|S'|}.
\]
For the other term, we calculate
\[
\frac{t\tk/(|S'|+2(1-p))}{\tk/|S'|} = t \cdot \frac{|S'|}{|S'| + 2(1-p)} \leq t.
\]
Therefore, we can bound the probability as
\begin{align*}
\Pr[|S \cap C_v| \in \{\tk-1, \tk\} \mid \{\ell_1, \ell_2\} \subseteq S]
  & \le
  \left(\left(1 - \frac{(3t/4-1)\tk/|S'|}{1-\tk/|S'|}\right)^{1-\tk/|S'|}
  t^{\tk/|S'|}\right)^{|S'|}\\
  &\leq
  \exp\left(-(3t/4-1-\log t)\tk\right) \leq \exp\left(-t\tk/2\right),
\end{align*}
for $t \geq \beta/p$ large enough.
We next observe that 
\[
\binom{|L|-2}{p|L|-2}/\binom{|L|}{p|L|}
= \frac{(|L|-2)!}{(p|L|-2)! (|L|-p|L|)!}\cdot
\frac{(p|L|)! (|L|-p|L|)!}{|L|!} =
\frac{p|L|}{|L|}\cdot \frac{p|L|-1}{|L|-1}
\leq p^2
\]
in order to conclude that 
\begin{multline}\label{equ:probV}
  \Pr[v \in V_{\tk}(S)] 
  =\left(\binom{|L|-2}{p|L|-2}/\binom{|L|}{p|L|}\right)\cdot
  \,\Pr\left[|S \cap C_v| \in \{\tk-1, \tk\} 
    \mid \{\ell_1, \ell_2\} \subseteq S \right]\\
  \leq 
p^2 \cdot \exp\left(-t\tk/2\right) .
\end{multline}
Now we can finally bound the expectation as follows:
\begin{align*}
  &\EX\left[
    \sum_{\substack{v \in V_{\tk}(S) \\ c_v \geq \beta \tk/p^2}} 
    c_v (\log c_v + \log^2 \tk + \log{(1/p)})
    \right]\\
  &= 
  \sum_{v \in \lev_{\geq \beta \tk/p^2}(L)} \Pr[ v \in V_{\tk}(S)]\, 
  c_v (\log c_v + \log^2 {\tk} + \log{(1/p)} )\\
\intertext{(grouping by level, using (\ref{equ:probV}),
and letting $l_c$ denote the number of vertices in $\lev_c(L)$)}
  &\leq 
  \sum_{c = \beta \tk/p^2}^{|L|} l_c\, p^2 
  e^{-cp/2}
  c (\log c + \log^2 \tk + \log{(1/p)} )\\
\intertext{(bounding the sum by an integral and using
$l_c = O(|L|c^{1/3})$~\cite{Dey98})}
    &\ll 
    \int_{c = \beta\tk/p^2}^{\infty} |L|c^{1/3} p^2 e^{-cp/2}\, 
    c (\log c + \log^2{\tk}  + \log{(1/p)} ) \,\mathrm{d}c\\
\intertext{(substituting $c = t\tk /p$ and using 
$\mathrm{d}c = (\tk/p)\mathrm{d}t$)}
&=
\int_{t = \beta/p}^{\infty} |L|(t\tk/p)^{1/3} p^2 e^{-t\tk/2}\, (t\tk/p) 
(\log (t/p) + \log^2{\tk} + \log{(1/p)}) (\tk/p)\,\mathrm{d}t \\
\intertext{(collecting the terms and simplifying) }
  &\ll |L|(\tk^3/p)   \int_{t = \beta/p}^{\infty} t^2 e^{-t\tk/2} \,\mathrm{d}t\\
\intertext{(solving the integral)}
  &= |L| (\tk^3/p)  (2\cdot (\beta/p)^2 \tk^{-1} + 8 (\beta/p) \tk^{-2} + 
  16 \tk^{-3}) e^{-\beta \tk/2p}\\
\intertext{(simplifying)}
  &\ll |L| e^{-\Theta(\beta)(\tk/p)},
\end{align*}
for $\beta$ large enough, as desired.
\end{proof}

\end{document}